\documentclass{article}

\usepackage{arxiv}

\usepackage[utf8]{inputenc} 
\usepackage[T1]{fontenc}    
\usepackage{hyperref}       
\usepackage{url}            
\usepackage{booktabs}       
\usepackage{amsfonts}       
\usepackage{nicefrac}       
\usepackage{microtype}      
\usepackage{lipsum}
\usepackage{lipsum}
\usepackage{amsmath} 
\usepackage{graphicx}
\usepackage{amsthm}
\usepackage{multirow}
\usepackage{amssymb} 
\graphicspath{ {./images/} }
\usepackage{color}
\usepackage{tikz}
\usetikzlibrary{arrows.meta,decorations.pathreplacing,positioning,calc}
\usepackage{pgfplots}
\pgfplotsset{compat=1.18}
\usepackage{float}
\usetikzlibrary{intersections} 
\usepackage{indentfirst}

\newtheorem{theorem}{Theorem}
\newtheorem{proposition}{Proposition}

\newtheorem{remark}{Remark}
\graphicspath{{figuras/},{logos/}}

\title{Impulsive Release Strategies for Wolbachia-Infected Mosquitoes under Temperature Induced Infection Loss}

\author{
   Jéssica C.S. Alves\textsuperscript{1}, \\
  University of São Paulo\\ Institute of Mathematics and Statistics\\ Department of Applied Mathematics\\ São Paulo, SP, 05508-090, Brazil.\\
  \texttt{alvesj@ime.usp.br} \\
   \And
  Christian E. Schaerer \\
  National University of Asunción\\ Polytechnic School\\ Campus UNA\\ San Lorenzo, Central, P.O. Box 2111 SL, Paraguay.\\
  \texttt{cschaer@pol.una.py} \\
  \And
 Claudia P. Ferreira \\
 São Paulo State University\\
 Institute of Biosciences \\ 
 Botucatu, SP, 18618-689, Brazil\\
    \texttt{claudia.pio@unesp.br}}
\begin{document}
\maketitle
\begin{abstract}
The release of \textit{Wolbachia}-infected mosquitoes is a promising strategy for controlling \textit{Aedes aegypti} populations, but exposure to high temperatures can induce temporary infection loss and compromise long-term persistence. In this work, we propose a population-dynamics model based on impulsive differential equations to describe the interaction between wild and infected mosquitoes, incorporating cytoplasmic incompatibility, periodic release interventions, and temperature-driven infection loss. Analytical threshold conditions are derived to characterize the existence and stability of periodic solutions associated with successful \textit{Wolbachia} establishment. Numerical simulations illustrate the theoretical results and enable a comparative analysis of the wMelPop, wMel, and wAlbB strains, highlighting how differences in thermal tolerance and fitness costs influence persistence after the release phase. The results emphasize the importance of accounting for environmental stress and impulsive interventions when designing effective and robust \textit{Wolbachia} release strategies.
\end{abstract}


\section{Introduction}
\footnotetext[1]{Corresponding author: alvesj@ime.usp.br.}
Mosquito population control is a common strategy for reducing the transmission of vector-borne diseases. In this work, we specifically consider the dynamics of \textit{Aedes aegypti}, the vector of several arboviruses, including dengue, Zika, chikungunya, and yellow fever \cite{ahebwa2023aedes1,zardini2024aedes3}. This species is now present in most tropical and subtropical regions and has the potential to establish itself in more than 80\% of countries across the globe \cite{cheong2025aedes4, lim2025overaedes5}. Current control methods include the use of insecticides, mechanical elimination of breeding sites, and biological control approaches  \cite{ dorigatti2018usingA5, lees2023insecticidesI1,onen2023mosquito}.

Biological control comprises several techniques, including the use of natural predators, parasitoids, and specific pathogens that affect mosquito development and survival. Among the most common approaches is the introduction of natural predators, such as certain fish species, including \textit{Gambusia affinis}, which consume mosquito larvae in aquatic environments \cite{steven2021predator}. Another strategy involves predatory mosquitoes and parasitoid insects: larvae of mosquitoes of the genus \textit{Toxorhynchites} prey on other mosquito larvae \cite{Toxorhynchites}, while parasitoid wasps, such as \textit{Trichogramma}, oviposit on mosquito eggs, contributing to population suppression \cite{parra2022insect}.
Additionally, insecticides remain an important component of mosquito control. However, large-scale application raises environmental concerns, including the evolution of insecticide resistance, which has been reported in several countries and compromises the effectiveness of mosquito-borne disease control efforts \cite{carson2023overcomingI3}.

In addition to these biological and chemical approaches, genetic control techniques, such as the Sterile Insect Technique (SIT), have been widely applied in the management of agricultural pests \cite{zheng2019incompatible} and adapted to mosquito sterilization \cite{almeida2022optimalSterile1,bliman2019implementationSterile2,huang2017modellingSterile7,huang2021impulsiveSterile5,li2020impulsiveSterile6}. Finally, an emerging biological control strategy involves the release of mosquitoes carrying \textit{Wolbachia}, a symbiotic bacterium that reduces mosquito vector competence and limits arbovirus transmission without requiring genetic modification of the mosquito population \cite{dorigatti2018usingA5,ross2021designing,pinto2021effectiveness,dos2022estimatingW1,zara2016estrategiasW2}.

The interaction between wild mosquitoes and those infected with \textit{Wolbachia} or genetically modified mosquitoes has been extensively investigated in the literature \cite{almeida2022optimalSterile1, bliman2019implementationSterile2, gonzales2025optimization, orozco2024comparing, almeida2019mosquitoW5, cardona2020WolbachiaW7, lopes2023exploringW4, yu2019modelingW8}. Most of this literature relies on mathematical modeling to analyze population dynamics, employing systems of ordinary differential equations \cite{almeida2022optimalSterile1, almeida2019mosquitoW5, campo2017optimalW10, dos2023establishingW11}, delay differential equations \cite{lopes2023exploringW4, benedito2020modelingD1, li2020impulsiveD2}, and discrete-time formulations \cite{blackwood2018cascadeA3, ufuktepe2022discreteA4, yu2019modelingDiscrete1}, among others. In addition, optimal control frameworks have been widely adopted to design efficient intervention strategies \cite{ almeida2022optimalSterile1, gonzales2025optimization, campo2017optimalW10,campo2018optimalW9, almeida2019optimal}. These modeling approaches aim to reduce the wild mosquito population and, consequently, mitigate disease transmission.

While much of the existing work focuses on continuous release strategies of \textit{Wolbachia}-infected mosquitoes \cite{cardona2020WolbachiaW7,lopes2023exploringW4,dos2023establishingW11, ogunlade2020modelingC1}, alternative formulations based on impulsive releases have also been proposed \cite{gonzales2025optimization, almeida2022vectorImpulsive1, li2024modelingImpulsive3, liu2023analysisImpulsive2, dianavinnarasi2021controllingImpulsive4}. Periodic and impulsive release strategies appear more frequently in the context of sterile insect techniques \cite{ huang2021impulsiveSterile5, li2020impulsiveSterile6} and in applications to other insect populations and pest control problems \cite{liu2023analysisImpulsive2, ALVES2026116517, pei2018optimizingPests1, wang2011analysisPests2}. These approaches are designed to optimize the timing and magnitude of releases, enhancing control efficiency and improving the feasibility of real-world interventions.

 This work specifically addresses the control of \textit{Aedes aegypti}, with a particular focus on the use of the endosymbiotic bacterium \textit{Wolbachia} as a biological control strategy. Several experimental and field studies have shown that the release of \textit{Wolbachia}-infected mosquitoes can substantially reduce vector competence and, consequently, the transmission of arboviruses \cite{ross2021designing,dos2022estimatingW1,zara2016estrategiasW2}. These findings support the use of \textit{Wolbachia} as an effective and potentially sustainable tool for the long-term control of \textit{Aedes aegypti} populations.

In this article, we model the interaction between wild \textit{Aedes aegypti} mosquitoes and mosquitoes carrying the \textit{Wolbachia} bacterium, incorporating the mechanism of cytoplasmic incompatibility, which reduces the reproductive success of uninfected females when mating with infected males \cite{kaur2024mechanism}. The model also includes two impulsive events: one representing the loss of \textit{Wolbachia} infection due to high-temperature exposure, and the other corresponding to the periodic release of \textit{Wolbachia}-infected mosquitoes. Within this framework of impulsive differential equations, we analyze the system dynamics to determine conditions that ensure the fixation of \textit{Wolbachia}-carrying mosquitoes in the population.

Field and laboratory evidence indicate that high temperatures can temporarily reduce \textit{Wolbachia} density and weaken cytoplasmic incompatibility in \textit{Aedes aegypti}, particularly in environments subject to intense thermal fluctuations \cite{ ross2017Wolbachia, gunasekaran2022sensitivity}. For the wMel strain, heatwaves have been associated with decreases in infection frequencies, with values falling to 83\% in larvae sampled directly from field habitats and to 88\% in eggs collected from ovitraps, before recovering toward baseline levels in subsequent months \cite{ross2020heatwaves}. Experimental studies also report drastic reductions in \textit{Wolbachia} density when immature stages experience elevated temperatures, with levels dropping to below 0.1\% of controls after seven days of heat exposure beginning at the egg stage, followed by partial recovery once temperatures return to milder ranges \cite{gunasekaran2022sensitivity,ulrich2016heat}. These observations motivate the inclusion in our model of a redistributive impulsive event representing episodic loss of infection due to high-temperature exposure, capturing sudden reductions in the proportion of infected mosquitoes, adding this proportion to wild mosquitoes, and their gradual reestablishment over time.

Temperature effects, however, are strain-dependent, with direct implications for the magnitude and frequency of the impulsive loss term in the proposed model. The wMel strain, while capable of blocking dengue and Zika transmission, shows sensitivity to elevated temperatures and exhibits reduced stability under heat stress \cite{ross2017Wolbachia,ulrich2016heat, ant2018wolbachia}. The wMelPop strain is even more susceptible, displaying pronounced fitness costs and sharp declines in density when exposed to high temperatures \cite{ross2019loss,ross2017Wolbachia}, suggesting a larger and more frequent infection-loss impulse. In contrast, the wAlbB strain demonstrates greater tolerance to thermal stress, maintaining comparatively higher densities under fluctuating or elevated temperatures \cite{gunasekaran2022sensitivity, mancini2021high}, implying a weaker impulsive reduction in infection levels. Since both pathogen-blocking capacity and thermal resilience influence persistence and replacement dynamics, we summarize these characteristics in Table \ref{tab:tab_1} to support the interpretation of the model parameters and their biological relevance.
\begin{table}[H]
\centering
{\footnotesize
\caption{Key biological traits of \textit{Wolbachia} strains: virus blocking, heat tolerance, and fitness effects.}\label{tab:tab_1}
\begin{tabular}{lccp{3.7cm}c}
\hline
\textbf{Strain} & \textbf{Virus block} & \textbf{Heat tolerance} & \textbf{Impact on host fitness} & \textbf{Ref.} \\
\hline
wMelPop & Very high & Low & Strong reduction in fecundity, egg viability, and longevity & \cite{ross2019loss,ross2017Wolbachia,ritchie2015application} \\
wMel    & Moderate  & Low & Moderate reduction in fecundity and egg viability; little effect on longevity & \cite{ulrich2016heat,mancini2021high, maciel2024Wolbachia} \\
wAlbB   & High      & Moderate/High & Mild effects; fertility may drop under heat stress & \cite{gunasekaran2022sensitivity,ant2018wolbachia,axford2016fitness} \\
\hline
\end{tabular}}
\end{table}

To address this problem, we formulate a population dynamics model governed by impulsive differential equations that incorporates two distinct impulsive mechanisms: the periodic release of \textit{Wolbachia}-infected mosquitoes and the temperature-driven loss of infection. This modeling approach extends existing impulsive frameworks by coupling operational interventions with environmentally induced perturbations. We establish analytical conditions for the existence, characterization, and stability of periodic solutions, and derive explicit threshold conditions that link the magnitude and timing of release events to the intensity and frequency of thermal infection loss. These results provide theoretical criteria guaranteeing long-term persistence of \textit{Wolbachia} infection in the mosquito population.

In addition to the theoretical analysis, we perform numerical simulations to validate the derived threshold conditions and to explore how episodic high-temperature events influence infection dynamics. The simulations illustrate the interplay between release schedules, environmental stress, and strain-specific biological traits, and identify scenarios in which maintaining infection persistence may require adjustments in release intensity or frequency.

One of the main contributions of this work is the derivation of an explicit analytical threshold for the number of \textit{Wolbachia}-infected mosquitoes that must be released at each impulsive intervention time 
$\tau$, based on the proposed impulsive population model. This threshold guarantees the successful establishment of the infected mosquito population and the elimination of the wild population in the absence of temperature-induced infection loss. Furthermore, when episodic infection loss due to high temperatures is incorporated into the model, the obtained condition ensures that the wild mosquito population can be maintained arbitrarily close to zero, provided that the release magnitude satisfies the derived threshold. These results provide a rigorous theoretical basis for designing impulsive release strategies capable of sustaining \textit{Wolbachia} infection under environmentally driven perturbations.

The remainder of this paper is organized as follows. Section 2 presents the interaction model between wild and \textit{Wolbachia}-infected \textit{Aedes aegypti} mosquitoes, establishing the fundamental properties of the solutions, including existence, uniqueness, positivity, boundedness, and the local stability of equilibrium states. Section 3 introduces the impulsive population-dynamic model and provides the theoretical analysis of the coupled impulsive system, where both release interventions and temperature-induced infection loss are formally described. In this section, we derive threshold conditions ensuring long-term infection persistence. Section 4 presents numerical simulations that support the analytical findings and assess the impact of episodic high-temperature events on infection dynamics. Finally, Section 5 summarizes the main conclusions and discusses implications for the design of impulsive release programs.


\section{Model Formulation}
We propose the following mathematical model to characterize the population dynamics of \textit{Aedes aegypti} mosquitoes, distinguishing between the wild population and the one infected with \textit{Wolbachia}:
\begin{equation}\label{eq:equation_1}
    \begin{cases}
        \dfrac{dS_1}{dt} = S_1\psi_1\left(1 - \dfrac{S_1 + S_2}{K} \right) \left(\dfrac{S_1+(1-\gamma)S_2}{S_1+S_2} \right) - \delta_1 S_1, \\[6pt]
        \dfrac{dS_2}{dt} = S_2\psi_2 \left(1 - \dfrac{S_1 + S_2}{K} \right) - \delta_2 S_2.
    \end{cases}
\end{equation}
with non-negative initial conditions and positive parameters, $S_1(t)$ represents the population of wild \textit{Aedes aegypti} mosquitoes, while $S_2(t)$ represents the population of \textit{Wolbachia}-infected mosquitoes over time $t$. The parameters $\psi_i$ and $\delta_i$ denote the birth and death rates of $S_1$ and $S_2$, respectively, for $i = 1,2$, whereas $\gamma \in [0,1]$ is the cytoplasmic incompatibility parameter. The parameter $K$ corresponds to the shared environmental carrying capacity of wild and infected populations.

The model in \eqref{eq:equation_1} was constructed to capture the main biological mechanisms governing the interaction between wild and \textit{Wolbachia}-infected \textit{Aedes aegypti} populations. The logistic-type terms reflect the limitation imposed by resource competition, assuming that both populations share the same environmental carrying capacity $K$. This shared constraint expresses the ecological competition for breeding sites and food resources, which are finite within the habitat.

The multiplicative structure of the growth terms ensures that population change is proportional to the current population size, following standard formulations in population dynamics. The additional factor
\[
\left(\frac{S_1+(1-\gamma)S_2}{S_1+S_2}\right)
\]
in the equation for $S_1$ represents the cytoplasmic incompatibility (CI) effect: when \textit{Wolbachia}-infected males mate with wild females, a fraction $\gamma$ of the eggs fail to develop \cite{kaur2024mechanism}. Thus, $(1-\gamma)$ quantifies the viable offspring fraction produced in mixed matings, reducing the reproductive success of the wild population.

The infected population, on the other hand, does not experience this reduction, as \textit{Wolbachia} transmission is assumed to be maternal and, in this case, fully efficient. Therefore, the growth of $S_2$ is described by a simpler logistic-like term, driven by $\psi_2$ and limited by the same carrying capacity $K$. Maternal transmission of \textit{Wolbachia} is assumed to be perfect, and therefore, vertical transmission failure is not considered. Possible mechanisms leading to infection loss are introduced and discussed in the following section.

For this model, we consider the following assumptions:
\begin{equation}\label{eq:equation_2}
    \psi_1 > \delta_1 \quad \text{and} \quad \psi_2 > \delta_2,
\end{equation}
which guarantees that birth rates exceed death rates, allowing population persistence in the absence of competition. Additionally, we assume
\begin{equation}\label{eq:equation_3}
    \psi_2 < \psi_1, \quad \text{and} \quad \delta_2 > \delta_1,
\end{equation}
implying that wild mosquitoes have greater fitness than \textit{Wolbachia}-infected mosquitoes. This assumption is biologically realistic, as infection typically imposes a fitness cost on mosquitoes, reducing reproductive rates and increasing mortality. However, depending on the \textit{Wolbachia} strain, these effects may vary significantly: some strains cause substantial reductions in mosquito fitness, while others exhibit milder impacts or even confer slight advantages under specific environmental conditions.

The model formulation thus reflects a balance between ecological competition and cytoplasmic incompatibility. The infection’s potential to replace the wild population depends on the interplay between the fitness cost and the reproductive advantage conferred by CI. Under suitable conditions, the system in \eqref{eq:equation_1} can exhibit coexistence or complete replacement of the wild population by the infected one, which will be explored in the subsequent analysis.

\begin{theorem}[Existence, Uniqueness, Positivity, and Global Boundedness]\label{thm:existence_uniqueness}
Consider the system \eqref{eq:equation_1}, with non-negative initial conditions, and assume that all parameters $\psi_i$, $\delta_i$, $K$, and $\gamma$ are positive constants. Then, the following statements hold:
\begin{enumerate}
    \item[I.] There exists a unique local solution $(S_1(t), S_2(t))$ to system \eqref{eq:equation_1} for any given nonnegative initial condition $(S_1(0), S_2(0))$.
    
    \item[II.] The positive orthant $\mathbb{R}^2_+ = \{(S_1,S_2) \,|\, S_1 \ge 0,\, S_2 \ge 0\}$ is invariant under the flow of \eqref{eq:equation_1}. In particular, if $S_1(0), S_2(0) \ge 0$, then $S_1(t), S_2(t) \ge 0$ for all $t \ge 0$.
    
    \item[III.] The solutions of \eqref{eq:equation_1} are uniformly bounded for all $t \ge 0$.
    
    \item[IV.] Consequently, for any nonnegative initial condition, there exists a unique, positive, and bounded solution $(S_1(t), S_2(t))$ of system \eqref{eq:equation_1} defined for all $t \ge 0$.
    
\end{enumerate}
\end{theorem}

The proof of Theorem \ref{thm:existence_uniqueness} is given in the Appendix \ref{appendix:A}.

Having established the existence, uniqueness, and boundedness of nonnegative solutions for system~\eqref{eq:equation_1}, we now proceed to investigate its equilibrium points and their stability properties. This analysis provides essential insights into the long-term behavior of the interacting populations, allowing us to identify possible coexistence or extinction scenarios for the wild and \textit{Wolbachia}-infected mosquito populations under the given ecological assumptions. Moreover, we will show that whenever the coexistence equilibrium exists, it corresponds to a saddle point and is therefore unstable.

\begin{theorem}[Equilibrium Points and Stability]\label{thm:thm_1}
Consider system~\eqref{eq:equation_1} with non-negative initial conditions.  
Under assumptions~\eqref{eq:equation_2} and~\eqref{eq:equation_3}, define
\[
R_1 = \frac{\psi_1}{\delta_1}
\quad \text{and} \quad
R_2 = \frac{\psi_2}{\delta_2}.
\]
The system may exhibit up to three equilibrium points in the region $\mathbb{R}^2_+ \setminus \{(0,0)\}$, whose existence and local stability properties are described below:

\begin{itemize}
    \item \textbf{\textit{Wolbachia}-free equilibrium $E_1$}: 
    \begin{equation*}
        E_1 = \left(K\left(1 - \frac{1}{R_1}\right),\, 0\right),
    \end{equation*}
    which exists if $R_1 > 1$ and is locally asymptotically stable if
    \begin{equation*}
        R_1 > R_2.
    \end{equation*}

    \item \textbf{Wild-free equilibrium $E_2$}: 
    \begin{equation*}
        E_2 = \left(0,\, K\left(1 - \frac{1}{R_2}\right)\right),
    \end{equation*}
    which exists if $R_2 > 1$ and is locally asymptotically stable if
    \begin{equation*}
        R_2 > (1-\gamma)R_1.
    \end{equation*}

    \item \textbf{Coexistence equilibrium $E_3$}: 
    \begin{equation*}
        E_3 = (\hat{S}_1,\, \hat{S}_2),
    \end{equation*}
    where
    \begin{align*}
        \hat{S}_1 &= \left(1 - \frac{1}{R_2} \right)\frac{K}{\gamma} 
        \left( \frac{R_1}{R_2} - (1 - \gamma) \right),\\[4pt]
        \hat{S}_2 &= \left(1 - \frac{1}{R_2} \right)K - \hat{S}_1.
    \end{align*}
    The equilibrium $E_3$ exists if and only if
    \begin{equation*}
        1 - \gamma < \frac{R_1}{R_2} < 1,
    \end{equation*}
    and, whenever it exists, it is a saddle point (thus unstable).
\end{itemize}
\end{theorem}

The proof of Theorem \ref{thm:thm_1} is given in the Appendix \ref{appendix:B}.

The conditions derived above admit a clear biological interpretation. 
The threshold $R_1>1$ indicates that a wild mosquito must produce, on average, more than one wild mosquito per generation for the wild population to persist. 
Similarly, the condition $R_2>1$ establishes the minimum requirement for the persistence of the Wolbachia-infected mosquito population. 
The inequality $R_1>R_2$ reflects a higher fitness of the wild population compared to the infected population in the absence of cytoplasmic incompatibility. 
Finally, the term $(1-\gamma)R_1$ represents the reduction in the effective fitness of the wild population induced by cytoplasmic incompatibility.

In Figure~\ref{fig:fig_1}, we present the phase portraits of model~\eqref{eq:equation_1} for different initial conditions, illustrating the trajectories of $S_1$ and $S_2$ relative to the equilibrium points. 
The results show that the basin of attraction associated with equilibrium $E_2$, which corresponds to the elimination of wild mosquitoes and the persistence of \textit{Wolbachia}-infected individuals, varies according to the parameter set of each strain. 
In particular, the attraction region for the wMelPop strain is significantly smaller than that observed for the other two strains. 
This behavior is consistent with the biological characteristics of wMelPop, which is known to reduce key mosquito fitness traits, as reported in Table~\ref{tab:tab_1}.
\begin{figure}[!ht]
        \centering
        \includegraphics[width=15cm]{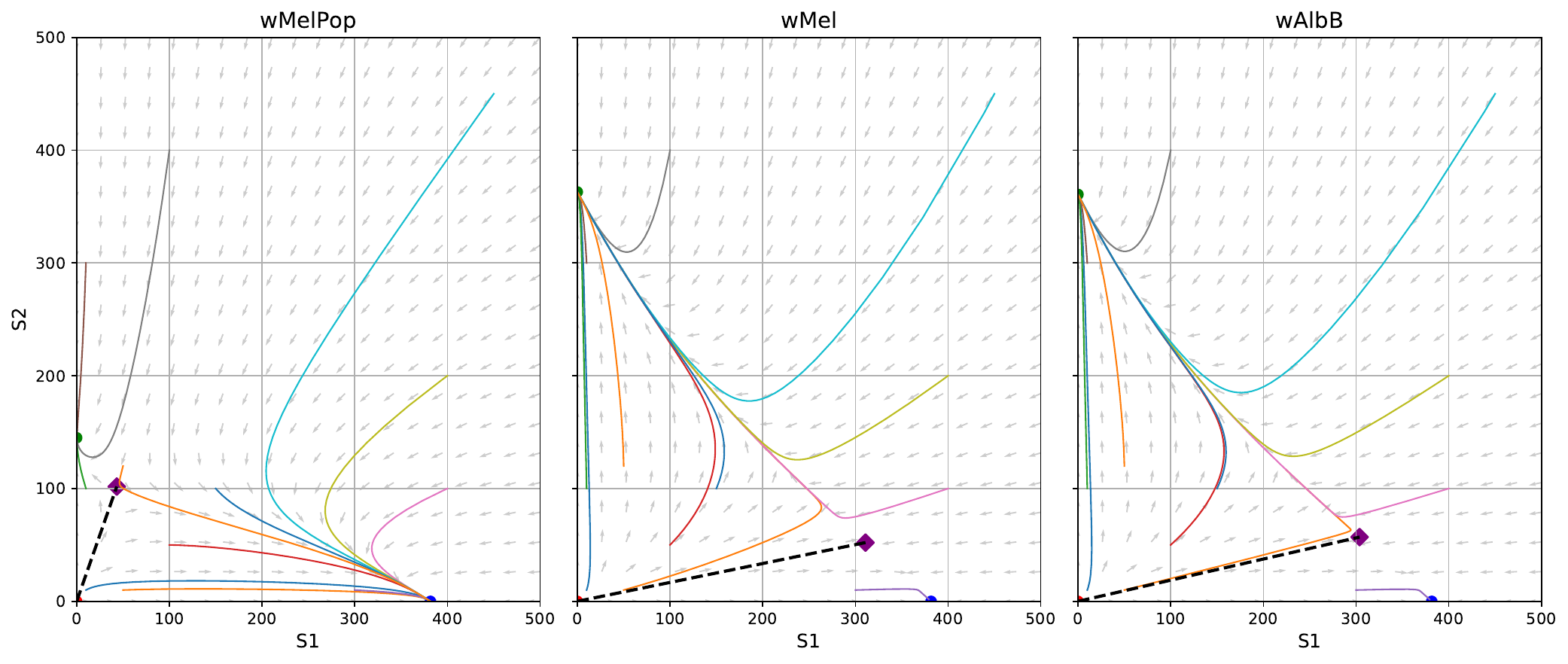}
        \caption{Phase portraits of system~\eqref{eq:equation_1} for $\gamma=0.95$, illustrating the qualitative behavior of trajectories for the three \textit{Wolbachia} strains. 
Different colors represent trajectories starting from distinct initial conditions. 
The equilibria $E_1$ (blue), $E_2$ (green), and $E_3$ (purple) are indicated. 
Parameter values are reported in Tables~\ref{tab:tab_3}--\ref{tab:tab_4}.}
\label{fig:fig_1}
\end{figure}

To conclude this section, we have established the main qualitative properties of the continuous model, including the existence, uniqueness, positivity, and boundedness of solutions, as well as the existence and local stability of its equilibrium points. In addition, the system behavior was illustrated through a parameter-space diagram and phase portraits, which together provide complementary insights into the role of model parameters and initial conditions. These results provide a solid theoretical foundation for extending the model to a more realistic framework that incorporates impulsive effects, which will be addressed in the next section.

\section{Impulsive Model}

Given our objective of studying and formulating release strategies for \textit{Wolbachia}-infected mosquitoes in a specific region, we adopt the framework of impulsive differential equations due to its practical relevance in modeling discontinuous interventions. In this model, we consider two types of periodic impulses: (i) the seasonal loss of infection caused by temperature variations and (ii) the impulsive release of infected mosquitoes.

We consider a fixed period $\tau>0$, corresponding to the time between successive mosquito interventions. First, we model the loss of infection caused by seasonal temperature variation through a redistributive impulse, since \textit{Wolbachia}-infected mosquitoes may lose the infection when exposed to high temperatures.At times $t = (n+l)\tau$, where $l \in [0,1]$ determines the relative timing of infection loss within each release period $\tau$, the seasonal loss event occurs within the release cycle and may coincide with a release when $l = 0$ or $l = 1$.
 The index $n \in \mathbb{Z}_+$ denotes the release cycle and organizes both types of impulses over time. At these moments, the population $S_2$ loses a proportion $\varphi \in [0,1]$ of individuals, which are transferred to the $S_1$ population:
\begin{equation*}
\begin{cases}
S_1(t^+) = S_1(t) + \varphi(t) S_2(t), \\
S_2(t^+) = (1 - \varphi(t)) S_2(t),
\end{cases}
\quad t = (n + l)\tau.
\end{equation*}

The loss rate $\varphi(t)$ is modeled as a smooth periodic function representing the seasonal variation in infection loss driven by temperature fluctuations:
\begin{equation}\label{eq:equation_5}
\varphi(t) = \varphi_{\max} \left[1 + \cos\left(\frac{2\pi (t - t_{\text{peak}})}{365}\right)\right]^2,
\end{equation}
where $\varphi_{\max}$ is related to the maximum seasonal loss rate, $365$ denotes the period of the seasonal cycle (one year in this case), and $t_{\text{peak}}$ represents the time of the year at which the loss is most intense. This formulation yields a smooth and periodic function satisfying $\varphi(t) \in [0, 4\varphi_{\max}]$. To ensure biological consistency, $\varphi_{\max}$ is chosen such that $\varphi(t) \in [0,1]$ for all $t$, which requires $\varphi_{\max} \in (0, \frac{1}{4}]$. Under this choice, the maximum proportion of infected individuals lost due to environmental effects never exceeds the total infected population.
\begin{figure}[h!]
    \centering
    \includegraphics[width=0.6\textwidth]{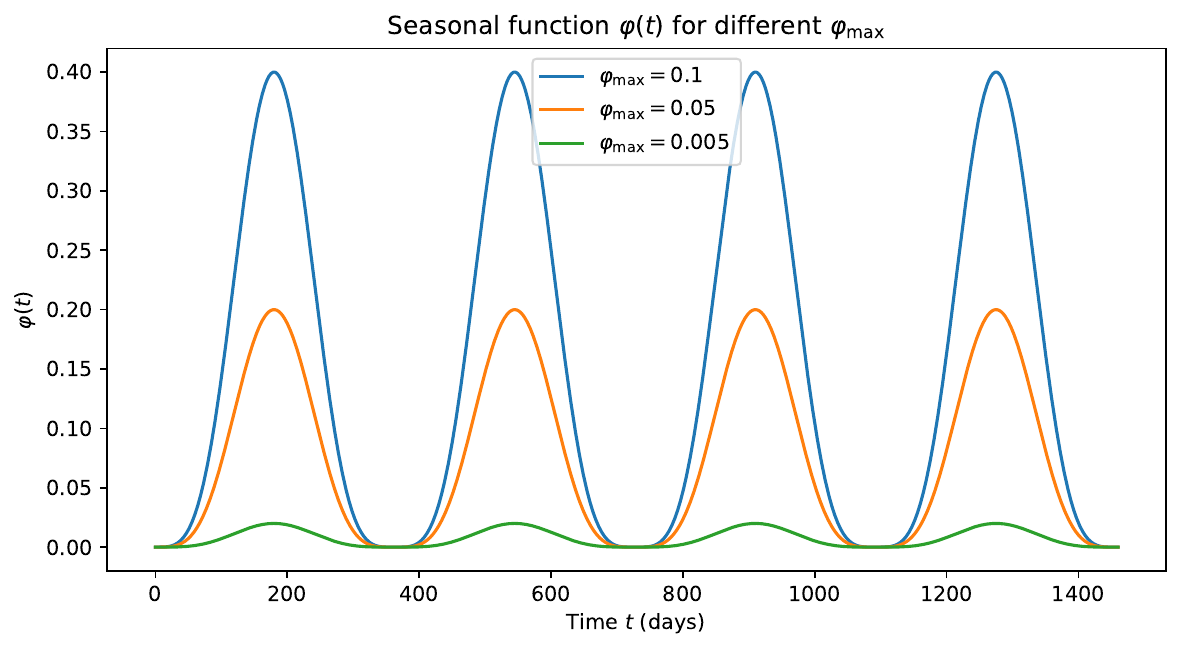}
    \caption{Seasonal function $\varphi(t)$ defined in \eqref{eq:equation_5} plotted for three different values of $\varphi_{\max}$.}
    \label{fig:fig_2}
\end{figure}

At each time instant $t = (n+1)\tau$, \textit{Wolbachia}-infected mosquitoes are released into the system in a quantity $u_n$. These impulsive releases are planned interventions aimed at increasing the proportion of infected individuals and suppressing the wild population, thereby promoting the spread and eventual fixation of the \textit{Wolbachia} infection. The release acts directly on the infected population $S_2$, while the wild population $S_1$ remains unchanged at the moment of intervention:
\begin{equation*}
\begin{cases}
S_1(t^+) = S_1(t), \\
S_2(t^+) = S_2(t) + u_n,
\end{cases}
\quad t = (n + 1)\tau.
\end{equation*}

The long-term dynamics of the system can be influenced by selecting $u_n$, $\tau$, and $\varphi_{max}$ in a way that favors the dominance of the infected population $S_2$ and suppresses the wild population $S_1$, which is particularly relevant for vector-borne disease management. Motivated by this control perspective, the complete model describing the interaction between \textit{Wolbachia}-infected and wild mosquito populations, accounting for both infection loss due to temperature and impulsive releases of infected individuals, is given by the following system of impulsive differential equations:

\begin{subequations}
\begin{align}
&\left.
\begin{cases}\label{eq:equation_6}
\dfrac{dS_1}{dt} = S_1\psi_1\left(1 - \dfrac{S_1 + S_2}{K} \right) \left(\dfrac{S_1+(1-\gamma)S_2}{S_1+S_2}\right) - \delta_1 S_1, \\
\dfrac{dS_2}{dt} = S_2\psi_2 \left(1 - \dfrac{S_1 + S_2}{K} \right) - \delta_2 S_2,
\end{cases}
\right.\hspace{-1cm}\text{if } t \neq (n + l)\tau, \text{ } t \neq (n + 1)\tau.\\
&\left.
\begin{cases}\label{eq:equation_7}
S_1(t^+) = S_1(t) + \varphi(t) S_2(t), \\
S_2(t^+) = (1 - \varphi(t)) S_2(t),
\end{cases}
\right.\text{if } t = (n + l)\tau.\\
&\left.
\begin{cases}\label{eq:equation_8}
S_1(t^+) = S_1(t), \\
S_2(t^+) = S_2(t) + u_n,
\end{cases}
\right.\text{if } t = (n + 1)\tau.
\end{align}
\end{subequations}
Within the impulsive differential equations framework adopted in this study, $\tau$  denotes the fixed time interval between successive releases of \textit{Wolbachia}-infected mosquitoes. Seasonal infection-loss events occur at shifted impulsive times of the form $(n+l)\tau$ and, depending on the value of $l$, may coincide with the release instants. In real-world applications, the release size $u_n$ is naturally constrained by the limited availability of \textit{Wolbachia}-infected mosquitoes. To reflect this practical limitation, the set of admissible releases is defined as $U:= \{ 0 \leq u_n\leq u_{\text{max}}, \mid n \geq 0 \}$, where $u_{\text{max}} \geq 0$ denotes the maximum number of \textit{Wolbachia}-carrying mosquitoes that can be introduced at any given intervention.

This formulation bridges theoretical modeling with real-world constraints, acknowledging that practical implementation must balance idealized strategies with biological and logistical limitations. The impulsive framework adopted here enables the integration of both environmental variability and active control strategies within a unified model. Seasonal loss of infection, caused by temperature fluctuations, is captured through redistributive impulses that reduce the number of infected individuals by transferring a proportion to the wild mosquito population. In contrast, periodic releases introduce \textit{Wolbachia}-infected mosquitoes into the environment and also act as interventions to sustain or increase their presence over time. This modeling approach provides a realistic and analytically tractable setting for exploring optimal release strategies and their interaction with adverse seasonal environmental stressors.

\subsection{Model behavior}
\begin{proposition}\label{prop:prop_1}
Let $(S_1(0), S_2(0))$ be a non-negative initial condition, $u_n \in U$, and $\varphi(t) \in [0,1]$. Let $(S_1(t), S_2(t))$ denote the solution of system \eqref{eq:equation_6}–\eqref{eq:equation_8}. Then, $(S_1(t), S_2(t))$ remains non-negative for all $t \geq 0$.
\end{proposition}

\begin{proof}
Between impulses, the solution $(S_1(t), S_2(t))$ remains non-negative for all $t \geq 0$, as established in Theorem~\eqref{thm:existence_uniqueness}-(I).

At impulse instants of the type $t = (n+l)\tau$, redistribution between the variables occurs, given by
\[
S_1(t^+) = S_1(t) + \varphi(t) S_2(t), \qquad 
S_2(t^+) = (1 - \varphi(t)) S_2(t).
\]
Since $\varphi(t) \in [0,1]$ and $S_2(t) \geq 0$, both $S_1(t^+)$ and $S_2(t^+)$ remain non-negative.

At impulse instants of the type $t = (n+1)\tau$, we have
\[
S_1(t^+) = S_1(t), \qquad S_2(t^+) = S_2(t) + u_n.
\]
Since $S_1(t) \geq 0$, $S_2(t) \geq 0$, and $u_n \geq 0$ by hypothesis, it follows that $S_1(t^+), S_2(t^+) \geq 0$.

Therefore, by induction over all impulse events and continuous intervals, the non-negativity of $S_1(t)$ and $S_2(t)$ is preserved for all $t \geq 0$.
\end{proof}

\begin{proposition}\label{prop:prop_2}
For each non-negative initial condition $(S_1(0), S_2(0))$, release amounts $u_n \in U$ ($u_n \geq 0$) at $t = (n+1)\tau$, and transfers $\varphi(t)$ defined by \eqref{eq:equation_5} at $t = (n+l)\tau$, the system \eqref{eq:equation_6}–\eqref{eq:equation_8} admits a unique solution defined on $[0, \infty)$.
\end{proposition}

\begin{proof}
From Theorem~\eqref{thm:existence_uniqueness}, the solution exists and is unique between impulses.

At discrete instants of the type $t = (n+l)\tau$, the functions $S_1(t)$ and $S_2(t)$ admit well-defined right-hand limits:
\[
S_1(k\tau^+) = \lim_{\epsilon \to 0^+} S_1(k\tau + \epsilon), 
\qquad 
S_2(k\tau^+) = \lim_{\epsilon \to 0^+} S_2(k\tau + \epsilon).
\]
At these moments, the impulsive transfers are given by
\[
S_1(t^+) = S_1(t) + \varphi(t) S_2(t), 
\qquad 
S_2(t^+) = (1 - \varphi(t)) S_2(t),
\]
where $\varphi(t)$ is defined by \eqref{eq:equation_5} and satisfies $\varphi(t) \in [0,1]$. These relations are linear and well defined, ensuring piecewise continuity of the solution and preservation of uniqueness after each impulse.

At instants of the type $t = (n+1)\tau$, the impulsive releases occur according to
\[
S_1(t^+) = S_1(t), 
\qquad 
S_2(t^+) = S_2(t) + u_n,
\]
with $u_n \in U$ and $u_n \geq 0$. This update is deterministic and also preserves the uniqueness of the solution.

Hence, the solutions are piecewise continuous, with well-defined left and right limits at every impulse time. By induction over the intervals $[n\tau, (n+1)\tau]$, the solution can be uniquely and consistently extended to the entire domain $[0, \infty)$.
\end{proof}

For the analysis of the model with impulsive interventions, we define the following auxiliary system. This system is adjustable, since both the releases $u_n \in U$ and the intervention period $\tau$ can be adjusted according to the chosen intervention strategy. Our approach is to derive an analytical solution by assuming the simplification $S_1(t) \approx 0$. Consequently, the system below is obtained from equations~\eqref{eq:equation_6}--\eqref{eq:equation_8} under the condition $S_1(t) = 0$:
\begin{align}
\begin{cases}\label{eq:equation_9}
\dfrac{dB}{dt}(t) = B\psi_2 \left(1 - \dfrac{B}{K} \right) - \delta_2 B,
\quad t \neq (n + l)\tau, \text{ } t \neq (n + 1)\tau,\\[6pt]   
B(t^+) = (1 - \varphi(t)) B(t),  \qquad t = (n + l)\tau,\\[6pt]
B(t^+) = B(t) + u_n, \qquad t = (n + 1)\tau.
\end{cases}
\end{align}

In what follows, we show the existence of a unique positive $\tau$-periodic solution for system~\eqref{eq:equation_9}, which characterizes the long-term periodic dynamics of the system under impulsive interventions.
\begin{theorem}\label{thm:thm_3}
Consider the auxiliary system \eqref{eq:equation_9}, with $\varphi(t) \in [0,1]$ defined by \eqref{eq:equation_5} and $u_n \in U$. Then, there exists a unique positive $\tau$-periodic solution $\bar{B}(t)$ given by:
\begin{align}\label{eq:equation_10}
\bar{B}(t) = \begin{cases}
\dfrac{\bar{K} B^+e^{r_2 (t-n\tau)}}{B^+\left(e^{r_2 (t-n\tau)} -1\right)+\bar{K}}, & n\tau < t \leq (n+l)\tau,\\
\dfrac{\bar{K} B^* e^{r_2 (t-(n+l)\tau)}}{B^*\left(e^{r_2 (t-(n+l)\tau)} -1\right)+\bar{K}}, & (n+l)\tau < t \leq (n+1)\tau,
\end{cases}
\end{align}
where
\begin{equation}\label{eq:equation_11}
\begin{cases}
B^{+}
    \;=\;
    \dfrac{\bar{K}\big((1-\varphi)e^{r_2\tau} - 1\big) + u_n A + \sqrt{\big(\bar{K}\big((1-\varphi)e^{r_2\tau} - 1\big) + u_n A\big)^2 + 4 u_n \bar{K} A}}{2A},\\
B^* = \dfrac{(1-\varphi)\bar{K}B^+e^{r_2l\tau}}{B^+\left(e^{r_2l\tau}-1\right)+\bar{K}},
\end{cases}
\end{equation}
with $A$ defined in \eqref{eq:A_def}.

Moreover, the solution $\bar{B}(t)$ is globally asymptotically stable.
\end{theorem}

The proof of Theorem \ref{thm:thm_3} is given in the Appendix \eqref{appendix:C}.

Establishing the uniform boundedness of the solutions is a crucial step in our analysis. Since we are dealing with population dynamics, it is natural to require that the trajectories remain biologically meaningful, that is, confined to a finite range for all time. Moreover, this property plays a central role in proving the global stability of the desired solutions, as it prevents the system from exhibiting unrealistic or uncontrolled behavior. Thus, after constructing and analyzing the auxiliary system, we now proceed to the proposition that ensures the uniform boundedness of the solutions of the impulsive model.

\begin{proposition}\label{prop:boundedness}
    Let $(S_1(t),S_2(t))$ be a solution of system \eqref{eq:equation_6}-\eqref{eq:equation_8}, with positive parameters, $\varphi(t) \in [0,1]$, $u_n \in U$ and non-negative initial conditions. Then $(S_1(t),S_2(t))$ is uniformly bounded.
\end{proposition}
\begin{proof}
From Theorem~\ref{thm:existence_uniqueness}-(III) we know that the solutions of
the continuous part of the model are uniformly bounded. In particular, there
exists a constant $M_c>0$ such that $S(t^-)=S_1(t^-)+S_2(t^-)\le M_c$ for every
interval between impulses and for every solution.

We now show that the impulsive effects do not destroy this boundedness. Consider
first the impulses occurring at times $t=(n+l)\tau$. At such instants,
\[
S_1(t^+) = S_1(t^-) + \varphi(t) S_2(t^-), \qquad 
S_2(t^+) = (1-\varphi(t)) S_2(t^-),
\]
with $0 \le \varphi(t) \le 1$. Hence,
\[
S(t^+) = S_1(t^+) + S_2(t^+) 
= S_1(t^-) + \varphi S_2(t^-) + (1-\varphi) S_2(t^-)
= S(t^-).
\]
Thus, impulses of this type preserve $S$ and do not increase its value.

Next, consider the impulses at times $t=(n+1)\tau$. In this case,
\[
S_1(t^+) = S_1(t^-), \qquad 
S_2(t^+) = S_2(t^-) + u_n,
\]
so that
\[
S(t^+) = S(t^-) + u_n.
\]
By assumption, the release magnitudes satisfy $u_n \le u_{\max}$ for each $n$.
Therefore,
\[
S(t^+) \le M_c + u_{\max}.
\]

Between two consecutive impulses, the solution evolves according to the continuous
dynamics and thus remains uniformly bounded. Consequently, defining
\[
M := M_c + u_{\max},
\]
we obtain $S(t) \le M$ for all $t \ge 0$ and for every solution of the full
impulsive system. Hence both $S_1(t)$ and $S_2(t)$ remain uniformly bounded.

Finally, since the vector field of the continuous dynamics is locally Lipschitz
and the jump maps either preserve $S$ or map bounded sets into bounded sets, no
finite-time blow-up can occur. Therefore, the solutions of the complete system is uniformly bounded.
\end{proof}

Having described the qualitative behavior of the model and the mechanisms governing the interaction between the two species, we now move to the formal analysis of the relevant equilibria in the impulsive regime. In particular, we examine the existence and stability of periodic solutions associated with the Wild-free state, in which the wild population of infected mosquitoes is eliminated, as well as the behavior of solutions that remain close to this state when small redistributive perturbations are introduced, namely, the impulses modeling the loss of infection due to high temperatures. In the next section, we establish conditions ensuring the existence of a positive, periodic, Wild-free solution and show that such a solution remains stable even in the presence of sufficiently small redistributive impulses. These results provide the necessary foundation for understanding how the system responds to perturbations and support the analysis of global stability around the dynamics we aim to achieve.

\subsection{Global Stability of Solutions}

In this section, we examine the existence and stability of the periodic solutions that constitute the main focus of this work, particularly the Wild-free solution, which represents the desired state in which the species 
$S_1$ is eliminated while the population of mosquitoes infected with \textit{Wolbachia} remains periodically stable. We first show that this solution exists uniquely and identify the conditions under which it is globally asymptotically stable in the absence of redistributive impulses, that is, when $\varphi(t) \equiv 0$. We then prove that, even in the presence of redistributive impulses of small magnitude, the system still admits a positive periodic solution close to 
$(0, \bar{S}_2)$, whose stability also remains preserved. These results provide the mathematical foundation for understanding how small perturbations affect the desired regime and ensure the robustness of the Wild-free solution in the impulsive setting.

From a biological perspective, establishing the existence and stability of these solutions is essential for assessing the feasibility of population management strategies. The Wild-free solution corresponds to the scenario in which the wild mosquito population is suppressed while the infected population persists sustainably over time, one of the central goals of this study. Demonstrating that this state not only exists but is also stable implies that moderate environmental fluctuations, seasonal variations, or small redistributive events do not compromise control of the wild population. Thus, the theoretical results indicate that the desired equilibrium is biologically attainable and remains robust under some realistic intervention conditions.

\begin{theorem}\label{thm:existence_eng}
Let $(S_1(0), S_2(0))$ be non-negative initial conditions. 
Assume that, in the absence of redistributive impulses (i.e., $\varphi \equiv 0$), 
system~\eqref{eq:equation_6}--\eqref{eq:equation_8} admits a unique, positive, 
$\tau$-periodic, wild-mosquito--free solution
\[
(0, \bar{S}_2(t)).
\]
Then, for $\varphi$ sufficiently small, the impulsive system admits a unique, positive, 
$\tau$-periodic solution
\[
(S_1^{\varphi}(t), \bar{S}_2^{\varphi}(t)),
\]
where $S_1^{\varphi}(t)$ remains uniformly small and 
$\bar{S}_2^{\varphi}(t)$ stays uniformly close to $\bar{S}_2(t)$ for all $t \ge 0$.
\end{theorem}

\begin{proof}
When $\varphi(t)\equiv 0$, setting $S_1(t)\equiv 0$ in system 
\eqref{eq:equation_6}--\eqref{eq:equation_8} reduces the equation for $S_2$ to the
auxiliary subsystem studied in~\eqref{eq:equation_9}. Since this subsystem admits a unique positive
$\tau$-periodic solution $\bar{B}(t)$ in Theorem~\ref{thm:thm_3}, we obtain the $\tau$-periodic solution
\[
(0,\bar{S}_2(t)), \qquad \text{with } \bar{S}_2(t)=\bar{B}(t),
\]
for the full system without redistributive impulses.

Now consider $\varphi \in (0,1]$. At each impulsive time $t=(n+l)\tau$, the redistribution
acts only by transferring a fraction $\varphi$ of $S_2$ into $S_1$. For $\varphi$
sufficiently small, this perturbation produces a small, continuous-in-$\varphi$ deviation
from the orbit $(0,\bar{S}_2(t))$. Between impulses the dynamics remain smooth and depend
continuously on initial data, so the solution
\[
(S_1^{\varphi}(t),\,\bar{S}_2^{\varphi}(t))
\]
exists for all $t\ge 0$ and stays close to $(0,\bar{S}_2(t))$. In particular,
$S_1^{\varphi}(t)$ remains uniformly small and the $S_2$–component satisfies
$\bar{S}_2^{\varphi}(t)\to \bar{S}_2(t)$ as $\varphi\to 0$. Hence, for $\varphi$
sufficiently small, the impulsive system admits a positive $\tau$-periodic solution that
is a perturbation of the unperturbed periodic orbit.
\end{proof}

\begin{theorem}[Global Stability]\label{thm:stability_sep}
Suppose the hypotheses Theorem~\ref{thm:existence_eng} hold and, moreover, that the
\(\tau\)-periodic wild mosquito free solution \((0,\overline{S}_2(t))\) satisfies the condition
\begin{equation}\label{globalstability}
    \overline{S}_2(t) \;>\; K\Big(1-\frac{\delta_1}{\psi_1}\Big)\qquad\text{for all } t\in\mathbb{R}.
\end{equation}
Then:
\begin{enumerate}
  \item[(i)] The periodic orbit \((0,\overline{S}_2(t))\) of the unperturbed system 
        (\(\varphi \equiv 0\)) is globally asymptotically stable in \(\mathbb{R}^2_+\).

  \item[(ii)] For \(\varphi\) sufficiently small, the unique \(\tau\)-periodic solution 
        \((S_1^{\varphi}(t), \overline{S}_2^{\varphi}(t))\) provided by 
        Theorem~\ref{thm:existence_eng} is asymptotically stable. 
        In particular, when \(\varphi\) is small, the component \(S_1^{\varphi}(t)\) remains 
        uniformly small for all \(t \ge 0\) and returns to a small neighborhood of zero after each 
        redistributive impulse, while \(\overline{S}_2^{\varphi}(t)\) stays close to 
        \(\overline{S}_2(t)\).
\end{enumerate}

\end{theorem}

\begin{proof}
\emph{(i) Stability of the wild–mosquito–free solution for $\varphi \equiv 0$.}

To prove this result, we employ an approach similar to that adopted in our previous work \cite{ALVES2026116517}.

From the second equation of the model, we obtain the estimate
\[
 \frac{dS_2}{dt}
 \;\le\;
 \psi_2\!\left(1-\frac{S_2}{K}\right)S_2 - \delta_2 S_2.
\]
Thus, we may apply the auxiliary system~\eqref{eq:equation_9} for comparison.  
By Theorem~\ref{thm:thm_3},
\[
B(t)\longrightarrow \bar{B}(t)
\qquad\text{as } t\to\infty.
\]
Hence, for any $\varepsilon>0$ sufficiently small there exists $t_1>0$ such that  
\[
B(t)<\bar{B}(t)+\varepsilon
\qquad\text{for all } t>t_1.
\]
By the Comparison Theorem (see, e.g., \cite{Laksh1989}, as well as our previous work \cite{ALVES2026116517}),

\begin{equation}\label{ineq:S2bound}
    S_2(t)\le B(t)<\bar{B}(t)+\varepsilon,
\qquad t>t_1.
\end{equation}

Next, using
\[
\frac{S_1+(1-\gamma)S_2}{S_1+S_2}\le 1,
\quad \gamma\in[0,1],
\]
we obtain
\[
\frac{dS_1}{dt}
\;\le\;
S_1\psi_1\!\left(1-\frac{S_2}{K}\right)-\delta_1 S_1.
\]
The equilibrium $S_1=0$ is stable whenever

\begin{equation}\label{ineq:condS1}
    \psi_1\!\left(1-\frac{S_2}{K}\right)<\delta_1.
\end{equation}

Using~\eqref{ineq:S2bound}, condition~\eqref{ineq:condS1} follows from
\[
\psi_1\!\left(1-\frac{\bar{B}(t)+\varepsilon}{K}\right)<\delta_1,
\qquad\forall t,
\]
which is   
\begin{equation*}
    \bar{B}(t)>K\left(1-\frac{\delta_1}{\psi_1}\right)=:\hat K.
\end{equation*}

Substituting \(S_2(t)\ge \hat K\) into the first equation of the system yields
\[
\frac{dS_1}{dt}
\le
S_1\psi_1\!\left(1-\frac{S_1+\hat K}{K}\right)-\delta_1 S_1.
\]
Consider the associated comparison system \(Z_1(t)\),
\begin{align*}
\begin{cases}
    \dfrac{dZ_1}{dt} = Z_1\psi_1\left(1-\dfrac{Z_1+\hat{K}}{K}\right)-\delta_1Z_1, \quad t \neq (n+l)\tau, \text{ } (n+1)\tau.\\
    Z_1(t^+) = Z_1(t)+\varphi(t)S_2, \quad t =(n+l)\tau.\\
    Z_1(t^+) = Z_1(t), \quad t =(n+1)\tau.
\end{cases}
\end{align*}
For which we know  
\(Z_1(t)\to 0\) as \(t\to\infty\).  
By comparison, for $\varepsilon>0$ sufficiently small there exists $t_2>t_1$ such that
\begin{equation}\label{ineq:S1small}
    S_1(t)\le Z_1(t)<\varepsilon,
\qquad t>t_2.
\end{equation}

Substituting~\eqref{ineq:S1small} into the second equation of the system~\eqref{eq:equation_6}-\eqref{eq:equation_8} and arguing as in  
Theorem~\ref{thm:thm_3}, we obtain the lower bound
\begin{equation}\label{ineq:s2bound2}
    \bar{B}(t)-\varepsilon \;\le\; S_2(t),
\qquad t>t_3>t_2.
\end{equation}

Combining~\eqref{ineq:S2bound}, \eqref{ineq:S1small}, \eqref{ineq:s2bound2} and letting $\varepsilon\to0$ yields
\[
S_1(t)\to0,
\qquad
S_2(t)\to\bar{B}(t)=:\overline{S}_2(t),
\qquad t\to\infty.
\]
Thus \((S_1(t),S_2(t))\to(0,\overline{S}_2(t))\) for every admissible initial condition.
Therefore the periodic orbit \((0,\overline{S}_2(t))\) is globally asymptotically stable in
\(\mathbb{R}^2_+\).

\emph{(ii) Stability for $\varphi$ small.}

By Theorem~\ref{thm:existence_eng}, for $\varphi$ sufficiently small the system admits a unique
$\tau$-periodic solution
\[
(S_1^\varphi(t),\,\overline{S}_2^\varphi(t))
\]
which depends continuously on $\varphi$ and satisfies
\[
(S_1^\varphi,\,\overline{S}_2^\varphi)\longrightarrow
(0,\,\overline{S}_2)
\qquad\text{as }\varphi\to0.
\]

Part \textit{(i)} shows that the orbit \((0,\overline{S}_2(t))\) attracts every solution of the
unperturbed system and that the attraction is uniform on bounded sets.
Since the vector field and the impulsive maps depend continuously on~$\varphi$,
the comparison arguments used in \textit{(i)} remain valid for $\varphi$ small:
all inequalities can be preserved by choosing $\varphi$ so that the perturbation terms are
bounded by an arbitrarily small $\varepsilon>0$.

Consequently,
\[
S_1^\varphi(t)\le \varepsilon
\qquad\text{for all }t\ge0
\]
and \(S_1^\varphi(t)\) returns to a neighborhood of zero after each impulse.
Similarly, the bounds for $S_2(t)$ obtained in \textit{(i)} persist under small perturbations, yielding
\[
\overline{S}_2(t)-\varepsilon
\;\le\;
\overline{S}_2^\varphi(t)
\;\le\;
\overline{S}_2(t)+\varepsilon.
\]

Thus the perturbed periodic orbit remains within an $\varepsilon$–tube of the unperturbed one
and attracts all nearby trajectories.  
Since $\varepsilon>0$ is arbitrary, the periodic solution
\((S_1^\varphi,\overline{S}_2^\varphi)\) is asymptotically stable for all sufficiently small $\varphi$.

\end{proof}

Theorem~\ref{thm:stability_sep} shows that, if the \(\tau\)-periodic profile of the infected
population \(\overline{S}_2(t)\) is sufficiently large (pointwise) relative to the threshold
\(K(1-\delta_1/\psi_1)\), then the wild population \(S_1\) cannot invade and the wild-free regime is
globally asymptotically stable. Moreover, this desirable regime is robust: small redistributive losses
of infection (modelled by \(\varphi\)) do not destroy the control outcome, which indicates the
biological feasibility and resilience of the proposed release strategy under moderate perturbations.

\subsection{A Method for Selecting $u_n$: A Sufficient Condition}\label{subs:subsec3.4}

Here, we use condition~\eqref{globalstability} to obtain a lower bound for $u_n$ that ensures stability. In other words, we determine a function $\eta(\tau)$ such that whenever the number of released infected mosquitoes satisfies $u_n \ge \eta(\tau)$, the solution remains stable around the wild-mosquito-free periodic orbit.

To achieve this, we need to guarantee that $\bar{S}_2(t) > \hat{K}$ for all $t \ge 0$. It is sufficient to require that the minimum value attained by $\bar{S}_2(t)$ within a single release cycle is greater than $\hat{K}$. We denote this minimum value by
\[
\bar{S}_2^{\min} = \min \bar{S}_2(t), \quad t \in (n\tau, (n+1)\tau].
\]
By examining the behavior of the solution over one period $\tau$, which is defined piecewise as
\begin{align*}
\bar{S}_2(t) =
\begin{cases}
\dfrac{\bar{K} B^+ e^{r_2 (t-n\tau)}}{B^+\left(e^{r_2 (t-n\tau)} - 1\right) + \bar{K}}, & n\tau < t \leq (n+l)\tau,\\[1em]
\dfrac{\bar{K} B^* e^{r_2 (t-(n+l)\tau)}}{B^*\left(e^{r_2 (t-(n+l)\tau)} - 1\right) + \bar{K}}, & (n+l)\tau < t \leq (n+1)\tau,
\end{cases}
\end{align*}
where $B^+$ and $B^*$ are given in~\eqref{eq:equation_11}, we find that the minimum value $\bar{S}_2^{\min}$ is attained at time $t = (n+l)\tau$. Therefore,
\[
\bar{S}_2((n+l)\tau) = \bar{S}_2^{\min} = Z^*.
\]

\subsubsection*{Derivation of the Sufficient Condition}

Since the minimum value of $\overline{S}_2(t)$ over each period is given by $Z^*$, imposing the threshold condition
\[
Z^* > \hat{K}
\]
and performing algebraic manipulations (similar to \cite{ALVES2026116517}) leads to a sufficient requirement on the release magnitude:
\[
u_n \ge \eta(\tau), \qquad \tau > 0,
\]
where the function $\eta(\tau)$ is
\begin{equation}\label{eq:equation_13}
\eta(\tau)
= \phi(\tau)\left(\dfrac{\phi(\tau)A - \bar{K}\big(e^{r_2 \tau}(1-\varphi)-1\big)}{\bar{K}+A\phi(\tau)}\right),
\qquad
\phi(\tau)
=\frac{\bar{K}\hat{K}}{\bar{K}(1-\varphi)e^{r_2 l\tau}-\hat{K}\big(e^{r_2 l\tau}-1\big)}.
\end{equation}

\subsubsection*{Establishing a Conservative Bound}

The function $\eta(\tau)$ may take negative values for certain $\tau>0$. Since biologically the release quantity must satisfy $u_n>0$, we restrict attention to values of $\tau$ for which $\eta(\tau)$ is positive. To obtain a uniform sufficient condition valid for all such $\tau$, we introduce the conservative requirement
\begin{equation}\label{eq:equation_14}
u_n > \max_{\tau>0} \eta(\tau) > 0.
\end{equation}
We now show that this maximum indeed exists.

\begin{proposition}\label{prop:propeta}
Let $\eta(\tau)$ and $\phi(\tau)$ be defined for $\tau\ge0$ by~\eqref{eq:equation_13}, with $A$ given by~\eqref{eq:A_def}, $0\le\varphi\le1$, and $l\in[0,1]$. Let $I$ be the maximal interval in $(0,\infty)$ on which the denominator of $\phi(\tau)$ is positive (so that $\phi>0$ on $I$). Then $\eta(\tau)$ attains a global maximum on $I$.
\end{proposition}

\begin{proof}
The functions $\phi$ and $\eta$ are smooth on $I$, hence $\eta$ is continuous and differentiable on $I$. Moreover,
\[
\lim_{\tau\to0^+}\phi(\tau)=\frac{\hat K}{1-\varphi}, \qquad \eta(0)\ \text{is finite}.
\]

If $I=[0,\infty)$, i.e., if $\bar K(1-\varphi)-\hat K>0$, then
\[
\phi(\tau)\sim C e^{-r_2 l\tau}\quad(C>0),\qquad
A(\tau)\sim (1-\varphi)e^{r_2\tau}\quad(\tau\to\infty),
\]
so
\[
\eta(\tau)\sim -C(1-\varphi)e^{r_2\tau(1-l)}\to -\infty\qquad(\tau\to\infty),
\]
because $1-l>0$.

If instead $I=[0,\tau_s)$, where the denominator of $\phi$ vanishes at some finite $\tau_s$, then a local expansion near $\tau_s$ shows that $\phi(\tau)\to+\infty$ as $\tau\to\tau_s^-$, and again $\eta(\tau)\to -\infty$ in this limit.

In both situations, $\eta$ is continuous on $I$, finite at $0$, and tends to $-\infty$ at the boundary of $I$. Hence there exists $M\in I$ such that $\eta(\tau)<\eta(0)$ for all $\tau\ge M$, and therefore $\eta$ attains its global maximum on the compact set $[0,M]\cap I$.
\end{proof}

With these results, we establish a uniform threshold for the release of \textit{Wolbachia}-infected mosquitoes, valid for any release period $\tau>0$. For values of $\tau$ for which $\eta(\tau)$ is positive, the function $\eta(\tau)$ itself provides a sufficient release quantity. In cases where $\eta(\tau)$ becomes negative, we instead use the global maximum of $\eta$, whose existence is guaranteed by Proposition~\ref{prop:propeta}. This maximum acts as a conservative bound that applies uniformly for all $\tau>0$.

In this way, for different parameter sets possibly varying according to the \textit{Wolbachia} strain considered we can determine a release quantity that maintains the infection over time while simultaneously driving the wild mosquito population to levels close to zero.

In the next section, we present numerical simulations that illustrate and validate the results obtained here. We also compare different \textit{Wolbachia} strains and analyze the influence of high temperatures on the loss of infection.

\section{Numerical Simulations}
This section presents the numerical results obtained for the three \textit{Wolbachia} strains considered in this study, namely wMelPop, wMel, and wAlbB, with the aim of assessing the feasibility and robustness of \textit{Wolbachia}-based mosquito control strategies under temperature-dependent infection loss. The numerical simulations allow us to evaluate how strain-specific biological traits influence invasion success, persistence, and long-term population outcomes in environmentally stressful conditions.

For all simulations, a release period of 365 days is considered, followed by an analysis of the post-release dynamics over a time horizon of 1095 days (three years), enabling the investigation of both short-term intervention effects and long-term post-release dynamics. The complete set of parameters employed in the simulations is provided in Tables~\ref{tab:tab_3} and~\ref{tab:tab_4}, with strain-specific parameter values summarized in Table~\ref{tab:tab_4}. These parameters ensure consistency with empirical studies and facilitate reproducibility.

\begin{table}[h!]
\centering
{\footnotesize
\caption{Multiplicative factors applied to the baseline demographic parameters of the wild mosquito population. 
The parameters for \textit{Wolbachia}-infected mosquitoes are computed as
\(
\psi_2 = \alpha_\psi \psi_1, \text{ and } \delta_2 = \alpha_\delta \delta_1,
\)
where $\alpha_\psi$ and $\alpha_\delta$ are the strain-specific multiplicative factors listed in this table.}
\label{tab:tab_2}
\begin{tabular}{lccc}
\hline
& \multicolumn{2}{c}{\textbf{Multiplicative factor}} & \\
\hline
\textbf{\textit{Wolbachia} strain} & \textbf{Birth rate $\alpha_\psi$} & \textbf{Mortality rate $\alpha_\delta$} & \textbf{References} \\
\hline
wMelPop & 0.50 & 1.50 & Assumed \\
wMel    & 0.95 & 1.10 & \cite{xue2018comparing} \\
wAlbB   & 0.85 & 1.00 & \cite{xue2018comparing} \\
Wild    & 1.00 & 1.00 & \cite{xue2018comparing} \\
\hline
\end{tabular}}
\end{table}

The analysis focuses on how strain-dependent differences in heat tolerance, virus-blocking efficiency, and fitness effects on the host mosquito are incorporated into the model through multiplicative factors applied to the demographic parameters of the wild mosquito population. These factors, reported in Table~\ref{tab:tab_2}, are derived from the biological characteristics listed in Table~\ref{tab:tab_1} and play a central role in determining the persistence of \textit{Wolbachia}-infected populations under temperature-induced infection loss.

The multiplicative factors associated with the wMel and wAlbB strains are taken directly from Table~3 of \cite{xue2018comparing}. In contrast, the values assumed for wMelPop reflect its well-documented severe fitness costs. Experimental studies report strong life-shortening effects and high egg mortality associated with this strain, resulting in substantially reduced survival and reproductive performance \cite{ritchie2015application, axford2016fitness}. These biological trade-offs are consistently discussed in both experimental and modeling studies comparing \textit{Wolbachia} strains and support the use of more extreme multiplicative factors for wMelPop in the absence of unified quantitative estimates \cite{ross2021designing, orozco2024comparing}.
\begin{table}[h!]
\centering
{\footnotesize
\caption{Common parameters used in the impulsive population-dynamic model.}
\label{tab:tab_3}
\begin{tabular}{lccc}
\hline
\textbf{Description} & \textbf{Parameter} & \textbf{Value} & \textbf{Range} \\
\hline
Relative timing of infection loss within the release period $\tau$ & $l$ & 0.5 & [0, 1] \\  
Time of the year corresponding to the temperature peak & $t_{\text{peak}}$ & 180 & -- \\
Carrying capacity for the mosquito population & $K$ & 500 & -- \\
\hline
\end{tabular}}
\end{table}

The numerical analysis is divided into two parts. The first part investigates the system dynamics during the release period and is used to validate the theoretical results derived in the previous section. Through numerical simulations, we show that releasing \textit{Wolbachia}-infected mosquitoes above the corresponding stability threshold leads to suppression of the wild mosquito population and the establishment of a stable infected population. To further support the robustness of the model, these simulations are performed under four distinct initial conditions, indicating that the observed dynamics and invasion outcomes are independent of the initial population configuration. In this phase, we also analyze the interactions among the release interval~$\tau$, the maximum infection-loss intensity~$\varphi_{\max}$, and the release amplitudes~$u_n$ for each parameter set, highlighting how these factors jointly affect invasion success.

The second part focuses on post-intervention dynamics, examining system behavior after the releases are halted and assessing how each \textit{Wolbachia} strain responds to environmental stress in the absence of ongoing control efforts. This analysis provides insights into the long-term sustainability of \textit{Wolbachia}-based interventions and the relative resilience of different strains under adverse temperature conditions.
\begin{table}[h!]
\centering
{\footnotesize
\caption{The values of model parameters for the wild mosquito population and for mosquitoes infected with different \textit{Wolbachia} strains.}
\label{tab:tab_4}
\begin{tabular}{lcccc}
\hline
\textbf{Description} & \textbf{Parameter} & \textbf{Value} & \textbf{Units} & \textbf{Range} \\
\hline
\multicolumn{5}{l}{\textbf{Wild mosquitoes $S_1$}} \\
Birth rate & $\psi_1$ & 0.300 & day$^{-1}$ & -- \\
Death rate & $\delta_1$ & 0.071 & day$^{-1}$ & -- \\
\hline
\multicolumn{5}{l}{\textbf{wMelPop-infected mosquitoes $S_2$}} \\
Birth rate & $\psi_2$ & 0.150 & day$^{-1}$ & -- \\
Death rate & $\delta_2$ & 0.106 & day$^{-1}$ & -- \\
CI probability & $\gamma$ & 0.90 & -- & $[0.9,\,1]$ \\
\hline
\multicolumn{5}{l}{\textbf{wMel-infected mosquitoes $S_2$}} \\
Birth rate & $\psi_2$ & 0.285 & day$^{-1}$ & -- \\
Death rate & $\delta_2$ & 0.0781 & day$^{-1}$ & -- \\
CI probability & $\gamma$ & 0.98 & -- & $[0.98,\,1]$ \\
\hline
\multicolumn{5}{l}{\textbf{wAlbB-infected mosquitoes $S_2$}} \\
Birth rate & $\psi_2$ & 0.255 & day$^{-1}$ & -- \\
Death rate & $\delta_2$ & 0.071 & day$^{-1}$ & -- \\
CI probability & $\gamma$ & 0.90 & -- & $[0.9,\,1]$ \\
\hline
\end{tabular}}
\end{table}

In particular, Table~\ref{tab:tab_4} reports the strain-specific parameters obtained by applying the multiplicative factors to the baseline demographic rates of the wild mosquito population \cite{xue2018comparing}. The parameter associated with cytoplasmic incompatibility is taken directly from \cite{orozco2024comparing}, ensuring consistency with empirical observations and enabling a coherent and biologically grounded comparison among the different \textit{Wolbachia} strains.

\subsection{Dynamics during the release phase}

\subsubsection*{Case 1: $\tau = 7$}

Figure~\ref{fig:fig_3} illustrates the population dynamics of wild and \textit{\textit{Wolbachia}}-infected mosquitoes under periodic releases every $\tau = 7$ days, using three parameter sets corresponding to the strains wMelPop (panel a), wMel (panel b), and wAlbB (panel c). In all cases, the release amplitudes satisfy the sufficient condition $u_n > \eta(7)$, ensuring that the infected population grows sufficiently to suppress the wild mosquitoes during the release period.

The simulations incorporate strain-specific infection-loss intensities, with $\varphi_{\max} = 0.0015$ for wMelPop, $\varphi_{\max} = 0.001$ for wMel, and $\varphi_{\max} = 0.0005$ for wAlbB, in agreement with empirical and modeling studies indicating that wMelPop is the most sensitive to elevated temperatures. Although the temperature-dependent function $\varphi(t)$ induces temporary increases in the wild population, these fluctuations remain limited and reflect only mild infection loss during temperature peaks.

Because wMelPop also imposes a strong fitness cost on the host mosquito, the corresponding threshold value is substantially higher, with $\eta(7) \approx 213$, compared to the smaller release amplitudes required for the other strains ($\eta(7) \approx 38$ for wMel and $\eta(7) \approx 33$ for wAlbB). Despite these differences, no significant loss of infection is observed for any of the three strains over the 365-day release period, highlighting the stabilizing effect of the intervention strategy.

\begin{figure}[H]
    \centering
    \includegraphics[width=1.0\linewidth]{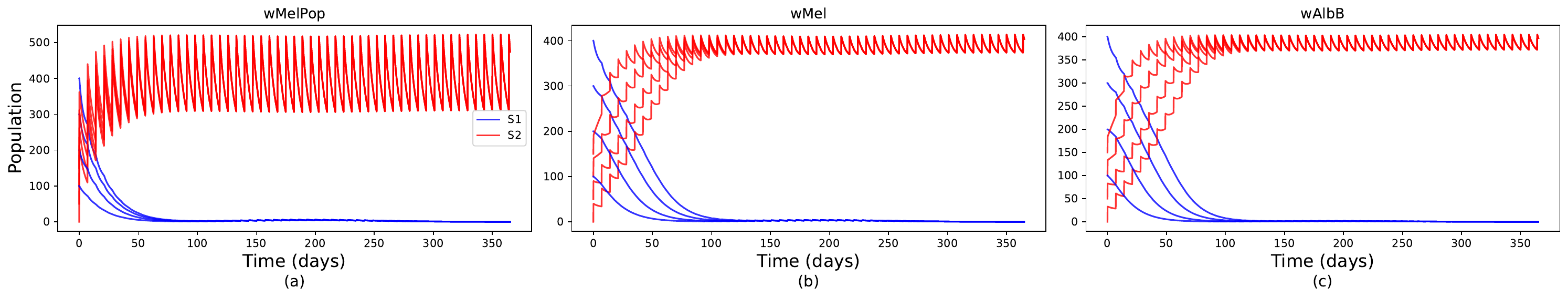}
    \caption{Wild and infected mosquito dynamics under periodic releases ($\tau = 7$ days) for the strains wMelpop (a), wMel (b), and wAlbB (c), simulated under four distinct initial conditions. All release amplitudes satisfy $u_n > \eta(7)$, ensuring dominance of the infected population. Infection-loss intensities are $\varphi_{\max} = 0.0015$, $0.001$, and $0.0005$ for the respective strains.}
 \label{fig:fig_3}
\end{figure} 
\subsubsection*{Case 2: $\tau = 7$ and increasing $\varphi_{\max}$}

Figure~\ref{fig:fig_4} considers the same release regime as in Case~1, but with increased intensities of temperature-induced infection loss, namely $\varphi_{\max} = 0.015,\; 0.01,\; 0.005$ for wMelPop, wMel, and wAlbB, respectively. Under these larger values, the impact of temperature becomes more pronounced, resulting in a faster decay of the infected population between releases.

These results highlight the critical role of the maximum loss rate in determining both the minimum release effort required and the overall feasibility of achieving population replacement under adverse thermal conditions.

\begin{figure}[H]
    \centering
    \includegraphics[width=1.0\linewidth]{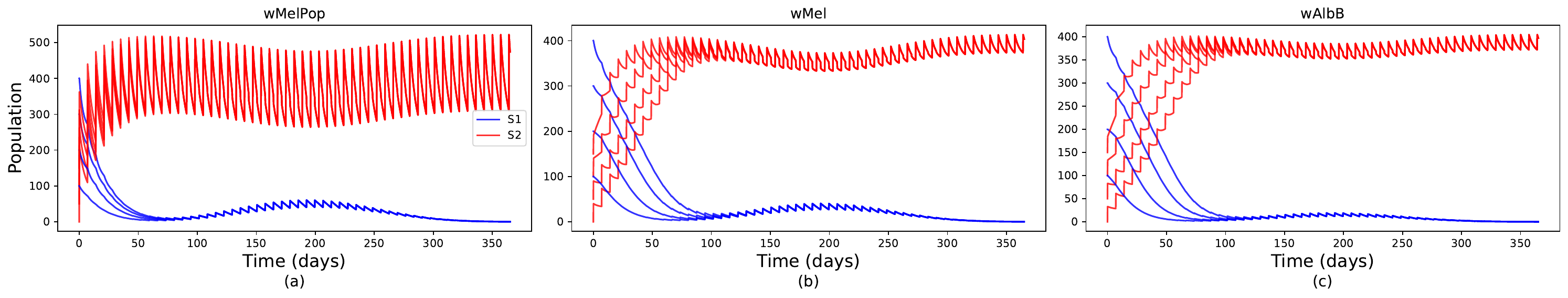}
   \caption{Dynamics of the wild and \textit{Wolbachia}-infected mosquito populations for the three \textit{Wolbachia} strains (wMelpop, wMel, and wAlbB), simulated under four distinct initial conditions and increased temperature-induced infection loss. The simulations consider $\tau = 7$ and higher loss rates, with $\varphi_{\max} = 0.015,; 0.01,; 0.005$ for wMelpop, wMel, and wAlbB, respectively. Larger values of $\varphi_{\max}$ make the thermal decay of infection more pronounced between releases, revealing the effect of stronger temperature stress on \textit{Wolbachia} persistence.}
\label{fig:fig_4}
\end{figure}

\subsubsection*{Case 3: $\tau = 14$}

Figure~\ref{fig:fig_5} presents the population dynamics obtained by increasing the release period to $\tau = 14$, while keeping the same infection-loss intensities used previously, namely $\varphi_{\max} = 0.0015$, $0.001$, and $0.0005$ for wMelPop, wMel, and wAlbB, respectively. Under this longer release interval, the release amplitudes $u_n$ must be increased to satisfy the sufficient condition $u_n > \eta(14)$. Compared with the results for $\tau = 7$ shown in Figure~\ref{fig:fig_3}, a more pronounced loss of infection is observed across all strains.

This behavior arises from the combined effect of a longer time between interventions and the temperature-dependent infection loss. As $\tau$ increases, the infected population remains exposed to thermal decay for a longer period before the next release, allowing infection loss to accumulate. Moreover, since the transition from infected to wild mosquitoes is proportional to the size of the infected population, higher infected densities lead to larger absolute losses during temperature peaks. Consequently, even for fixed values of $\varphi_{\max}$, increasing the release period amplifies the impact of thermal stress on the persistence of \textit{Wolbachia}.
\begin{figure}[H]
    \centering
    \includegraphics[width=1.0\linewidth]{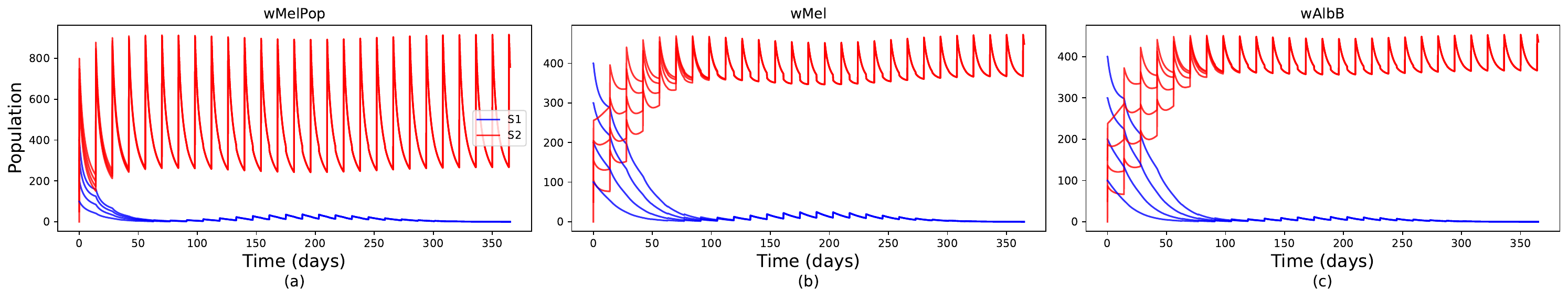}
    \caption{Dynamics of wild and \textit{Wolbachia}-infected mosquitoes, simulated under four distinct initial conditions, with higher infection-loss intensities: $\varphi_{\max} = 0.0015$ (wMelPop), $0.001$ (wMel), and $0.0005$ (wAlbB), and release interval $\tau = 14$. Compared with Figure~\ref{fig:fig_3}, temperature-induced infection loss becomes more apparent, although periodic releases continue to sustain the infected population.}
    \label{fig:fig_5}
\end{figure}

Having analyzed the population dynamics during the release period, we now turn to the second part of our numerical study, where we investigate what happens once the release interventions stop. This analysis allows us to evaluate the long-term persistence of \textit{Wolbachia} under temperature-dependent infection loss and to determine how each strain responds when external introductions are no longer provided.

\subsection{Dynamics after the releases phase}

What happens to the population dynamics once the releases of infected mosquitoes stop?

In Figures~\ref{fig:fig_6}--\ref{fig:fig_8}, we examine the system’s behavior over three years, assuming that \textit{Wolbachia}-infected females are released only during the first year. After this initial intervention phase, the system evolves without further release impulses, and the only remaining impulses are those associated with temperature-induced infection loss. For this analysis, we set $\tau = 14$ and explore three scenarios by varying the parameter $\varphi_{\max}$, which controls the intensity of thermal loss of infection.

\subsubsection*{Scenario 1: low infection loss}
 
In this first scenario, we consider $\varphi_{\max} = 0.0015$, $0.001$, and $0.0005$ for wMelpop, wMel, and wAlbB, respectively. As shown in Figure~\ref{fig:fig_6}, these values represent relatively weak temperature-induced loss. Under these conditions, once the releases end, the system maintains a stable configuration for all three \textit{Wolbachia} strains: the wild population remains suppressed, and the infected population stabilizes at a positive equilibrium. This indicates that, for sufficiently low thermal loss rates, the equilibrium achieved during the release period persists even in the absence of continued interventions, despite the seasonal temperature fluctuations encoded in $\varphi(t)$.
\begin{figure}[H]
    \centering
    \includegraphics[width=1.0\linewidth]{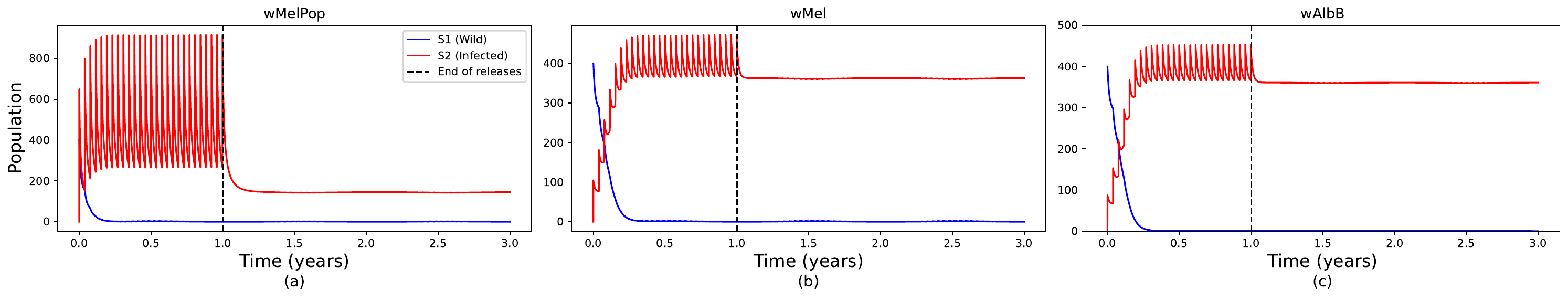}
    \caption{Population dynamics of wild and \textit{Wolbachia}-infected mosquitoes under low temperature-induced infection loss after releases stop. The simulations consider $\tau = 14$ and $\varphi_{\max} = 0.0015$, $0.001$, and $0.0005$ for wMelPop, wMel, and wAlbB, respectively. Releases occur only during the first year, and thereafter the system evolves under temperature-driven impulses alone. All three strains maintain a stable infected population and suppress the wild population throughout the three years.}
 \label{fig:fig_6}
\end{figure}
\subsubsection*{Scenario 2: moderate infection loss}  
In this scenario, we increase the infection-loss intensities to $\varphi_{\max} = 0.015$, $0.01$, and $0.005$ for wMelPop, wMel, and wAlbB, respectively. As shown in Figure~\ref{fig:fig_7}, the stronger temperature-induced loss leads to a more noticeable oscillatory pattern in the population trajectories, reflecting the increased periodicity driven by thermal stress. Despite these larger fluctuations, the system still converges to a stable configuration over the three-year window: the wild population remains controlled, and the infected population persists at positive levels for all three strains. This indicates that, under moderate thermal loss, the impulsive releases carried out during the first year are sufficient to sustain long-term \textit{Wolbachia} establishment.

\subsubsection*{Scenario 3: high infection loss}  
In the third scenario, we further increase the temperature-induced loss to $\varphi_{\max} = 0.15$, $0.1$, and $0.05$ for wMelPop, wMel, and wAlbB, respectively. Figure~\ref{fig:fig_8} reveals the most pronounced differences among the three strains. For wMelPop and wMel, the high infection-loss rates prevent the infected population from remaining established after the release period, leading to a rapid decline of \textit{Wolbachia} prevalence and a subsequent recovery of the wild population. In contrast, wAlbB remains substantially more resilient: although the infected population experiences stronger seasonal fluctuations, it still maintains a stable positive level throughout the post-release years. This outcome underscores the strong influence of thermal tolerance on long-term infection persistence once releases cease.

\begin{figure}[H]
    \centering
    \includegraphics[width=1.0\linewidth]{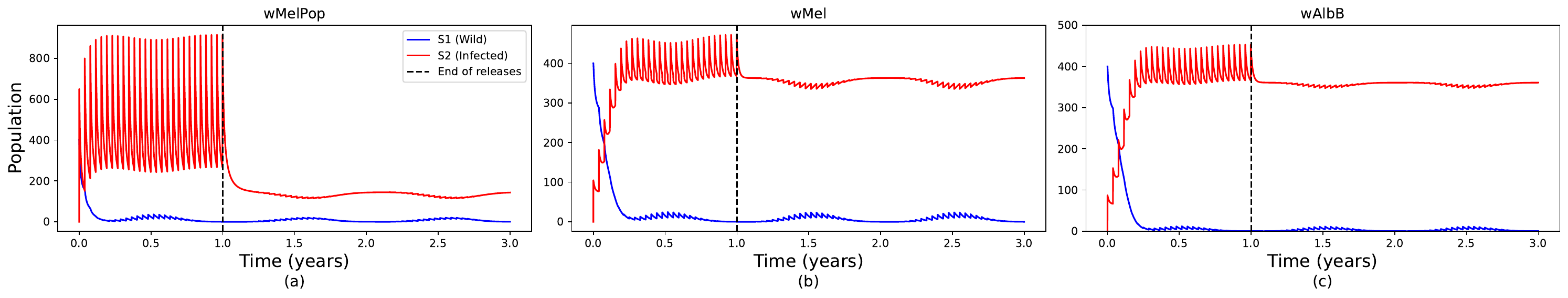}
   \caption{Population dynamics under moderate temperature-induced infection loss with $\tau = 14$ and $\varphi_{\max} = 0.015$, $0.01$, and $0.005$ for wMelPop, wMel, and wAlbB, respectively. After the end of releases (year 1), the system exhibits stronger periodic fluctuations due to increased thermal decay of infection. Despite these oscillations, all three strains remain stable and continue to suppress the wild population over the three-year horizon.}
 \label{fig:fig_7}
\end{figure}
Across the three scenarios, the long-term persistence of \textit{Wolbachia} after the end of the releases is strongly shaped by the magnitude of temperature-induced infection loss. Under low and moderate thermal stress (Scenarios~1 and~2), all three strains maintain a stable infected population and keep the wild population suppressed throughout the three-year period. However, when the loss rate becomes high (Scenario~3), the differences between strains become pronounced: wMelPop and wMel are unable to withstand the strong thermal pressure and rapidly lose their infected populations once releases stop, whereas wAlbB remains resilient and maintains infection persistence. These results highlight the combined importance of thermal tolerance and post-release dynamics, emphasizing that strain choice plays a decisive role in determining whether \textit{Wolbachia} can remain established without continued intervention.

\begin{figure}[H]
    \centering
    \includegraphics[width=1.0\linewidth]{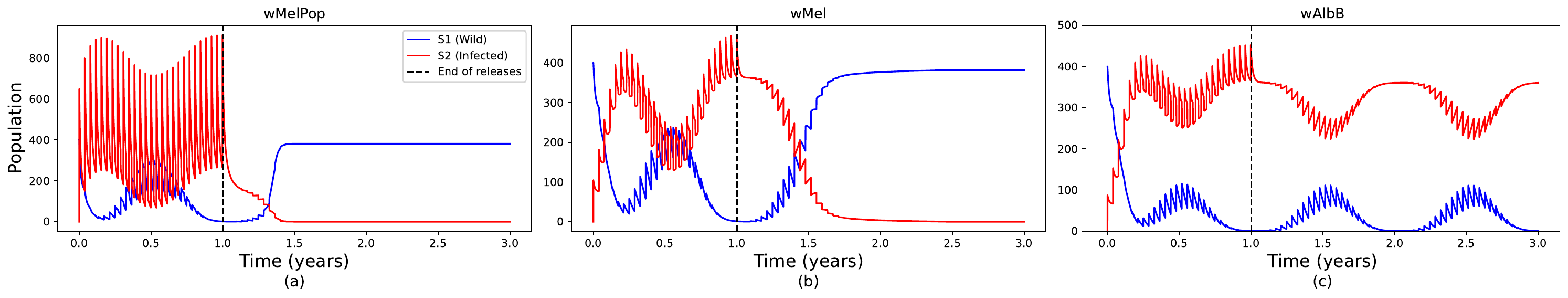}
    \caption{Population dynamics under high temperature-induced infection loss with $\tau = 14$ and $\varphi_{\max} = 0.15$, $0.1$, and $0.05$ for wMelPop, wMel, and wAlbB, respectively. In this scenario, wMelPop and wMel fail to maintain \textit{Wolbachia} infection after releases cease, leading to loss of the infected population and resurgence of the wild mosquitoes. In contrast, wAlbB remains stable despite stronger seasonal fluctuations, demonstrating its superior thermal tolerance.}
\label{fig:fig_8}
\end{figure}

Our numerical findings are also consistent with field and laboratory observations reported in the literature, which show that exposure to high temperatures can temporarily reduce \textit{Wolbachia} density and weaken cytoplasmic incompatibility, particularly for heat-sensitive strains such as wMel. Empirical studies document substantial drops in infection frequencies during heatwaves and severe reductions in bacterial density when immature stages develop under elevated temperatures, followed by gradual recovery once conditions normalize \cite{ross2019loss,gunasekaran2022sensitivity,ross2017Wolbachia}. This behavior, characterized by episodic loss of infection accompanied by post-stress reestablishment, motivated the redistributive impulsive event incorporated into our model. Across Scenarios 1–3, our simulations reproduce this pattern: when thermal stress is low or moderate, all strains recover after temperature-induced losses, whereas under intense stress, only wAlbB can regain and maintain high infection levels. This agreement between empirical evidence and model outcomes reinforces the relevance of including temperature-driven infection loss in predicting long-term \textit{Wolbachia} persistence.

\section{Final Remarks}

This work combined analytical results and numerical simulations to investigate impulsive release strategies of \textit{Aedes aegypti} mosquitoes infected with \textit{Wolbachia} under temperature-dependent infection loss. The results showed that long-term infection persistence depends on the interaction between strain-specific biological characteristics, thermal stress intensity, and operational release parameters.

The simulations confirmed the theoretical threshold conditions for the establishment of the infected population. Although all strains can suppress the wild population when release amplitudes exceed the critical threshold, biological differences significantly influence the practical feasibility of release strategies. In particular, the wMelPop strain presented higher fitness costs and required more intense releases, whereas wMel and, especially, wAlbB achieved population replacement with lower release effort.

Furthermore, increasing the release interval and the infection-loss intensity was shown to increase the difficulty of maintaining infection, highlighting the importance of jointly considering environmental and operational factors when planning interventions. After the end of releases, only strains with greater thermal tolerance, such as wAlbB, were able to maintain infection under conditions of higher thermal stress.

Overall, the results indicate that impulsive release strategies can successfully establish \textit{Wolbachia} in \textit{Aedes aegypti} populations, provided that strain selection, release planning, and environmental conditions are properly considered. It is important to emphasize, however, that the proposed mathematical model represents a simplified description of reality, based on assumptions designed to capture the main biological and environmental mechanisms involved in the process. Nevertheless, the obtained results provide relevant insights into the system dynamics and may contribute to the planning and evaluation of intervention strategies.

Future work may include the incorporation of spatial heterogeneity, stochastic temperature variations, and the development of optimal control approaches to further improve the design of release strategies.

\appendix
\section*{Appendix}
\addcontentsline{toc}{section}{Appendix}
\section{Proof of Theorem \ref{thm:existence_uniqueness}}\label{appendix:A}
\begin{proof}
To prove the theorem, we proceed in some steps.

\medskip
\noindent{(I) Local Lipschitz continuity.}

Define
\begin{align*}
f_1(S_1,S_2) &= S_1\psi_1\left(1 - \dfrac{S_1+S_2}{K}\right)\left(\dfrac{S_1+(1-\gamma)S_2}{S_1+S_2}\right) - \delta_1 S_1, \\
f_2(S_1,S_2) &= S_2\psi_2\left(1 - \dfrac{S_1+S_2}{K}\right) - \delta_2 S_2.
\end{align*}

Let $f = (f_1,f_2)^\top$ and consider the domain 
\[
D = \{(S_1,S_2) \in \mathbb{R}^2_{+} : S_1 + S_2 > 0\}.
\]
In $D$, both $f_1$ and $f_2$ are rational functions formed by sums, products, and quotients of polynomials. Hence, they are $C^\infty$ on $D$.  

The Jacobian matrix of $f$ is given by

\begin{align}\label{jacobi1}
    J(S_1,S_2) =
\begin{bmatrix}
\dfrac{\partial f_1}{\partial S_1} & \dfrac{\partial f_1}{\partial S_2} \\[6pt]
\dfrac{\partial f_2}{\partial S_1} & \dfrac{\partial f_2}{\partial S_2}
\end{bmatrix},
\end{align}

where
\begin{align}\label{jacobi2}
\frac{\partial f_1}{\partial S_1} &= 
\psi_1\left(1 - \frac{S_1 + S_2}{K}\right)\left[\left(\frac{S_1 + (1 - \gamma)S_2}{S_1 + S_2}\right) + \frac{\gamma S_1 S_2}{(S_1 + S_2)^2}\right]
- \frac{\psi_1 S_1}{K}\left(\frac{S_1 + (1 - \gamma)S_2}{S_1 + S_2}\right) - \delta_1, \\\nonumber
\frac{\partial f_1}{\partial S_2} &= 
S_1\psi_1\left[-\frac{1}{K}\left(\frac{S_1 + (1 - \gamma)S_2}{S_1 + S_2}\right)
- \left(1 - \frac{S_1 + S_2}{K}\right)\frac{\gamma S_1}{(S_1 + S_2)^2}\right], \\\nonumber
\frac{\partial f_2}{\partial S_1} &= -\frac{\psi_2 S_2}{K}, \\\nonumber
\frac{\partial f_2}{\partial S_2} &= \psi_2\left(1 - \frac{S_1 + 2S_2}{K}\right) - \delta_2.
\end{align}

All entries of $J(S_1,S_2)$ are continuous in $D$, so $f$ is locally Lipschitz on every compact subset of $D$.  
Therefore, by the Picard–Lindelöf theorem, for any initial condition $(S_1(0), S_2(0)) \in D$, there exists a time $t_{\max} > 0$ and a unique solution $(S_1(t), S_2(t))$ defined on $[0, t_{\max})$.

\medskip
\noindent{(II) Positive invariance of $\mathbb{R}_+^2$.}

If $S_1(0) = 0$ and $S_2(0) \ge 0$, then $f_1(0,S_2) = 0$, which implies $\frac{dS_1}{dt}(0) = 0$ and thus $S_1(t) \ge 0$ for $t > 0$.  
Similarly, if $S_2(0) = 0$ and $S_1(0) \ge 0$, then $f_2(S_1,0) = 0$, implying $\frac{dS_2}{dt}(0) = 0$.  
Hence, trajectories starting in $\mathbb{R}_+^2$ remain there, since the flow is tangent or directed inward along the coordinate axes.  
Therefore, the nonnegative orthant $\mathbb{R}_+^2$ is positively invariant.

\medskip
\noindent{(III) Boundedness of solutions.}

Let $S(t) = S_1(t) + S_2(t)$. Then,
\begin{align*}
\frac{dS}{dt} &= f_1 + f_2 \\
&= S_1\psi_1\left(1 - \frac{S}{K}\right)\left(\frac{S_1 + (1-\gamma)S_2}{S_1 + S_2}\right)
+ S_2\psi_2\left(1 - \frac{S}{K}\right)
- \delta_1 S_1 - \delta_2 S_2.
\end{align*}
Since $0 \le \frac{S_1 + (1-\gamma)S_2}{S_1 + S_2} \le 1$, we obtain
\[
\frac{dS}{dt} \le \psi S\left(1 - \frac{S}{K}\right) - \delta S,
\]
where $\psi = \max\{\psi_1, \psi_2\}$ and $\delta = \min\{\delta_1, \delta_2\}$.

To analyze this inequality, consider the comparison system
\begin{align*}
\frac{dy}{dt} &= \psi y\left(1 - \frac{y}{K}\right) - \delta y,\\
y(0) &= S(0).
\end{align*}
This equation can be written as
\[
\frac{dy}{dt} = y(r - a y),
\quad \text{where } r = \psi - \delta, \quad a = \frac{\psi}{K}.
\]
This is a logistic equation with an analytical solution
\[
y(t) = \frac{y(0)r e^{rt}}{a y(0)(e^{rt} - 1) + r},
\]
which satisfies $0 < y(t) < K\left(1 - \frac{\delta}{\psi}\right) =: y^*$ for all $t > 0$ provided $\psi > \delta$.
By the Comparison Lemma, $S(t) \le y(t)$ for all $t$ in the domain of existence.
Hence, $S(t)$ is uniformly bounded by $\max\{S(0), y^*, K\}$. Finally, since $S(t)$ remains uniformly bounded for all $t \ge 0$, the solution
$(S_1(t),S_2(t))$ stays in a compact subset of $\mathbb{R}^2_+$. Because the
vector field is locally Lipschitz, no finite–time blow-up can occur, and the
solution can be extended indefinitely. Therefore, $(S_1(t), S_2(t))$ exists for
all $t \ge 0$.

\noindent{(IV) Solution for all $t\geq 0$.}

From the above results, for any nonnegative initial condition, there exists a unique global solution $(S_1(t),S_2(t)) \in \mathbb{R}^2_+$ for all $t \ge 0$,
and both components remain bounded.
\end{proof}

\begin{remark}
    \medskip
\noindent{Behavior at the origin $(0,0)$.}

The function
\[
g(S_1,S_2) = \frac{S_1 + (1-\gamma)S_2}{S_1 + S_2}
\]
is undefined at $(0,0)$. However, its limits along the coordinate axes are
\[
\lim_{S_2 \to 0^+} g(S_1,S_2) = 1, \qquad
\lim_{S_1 \to 0^+} g(S_1,S_2) = 1 - \gamma.
\]
Thus, $g$ admits a continuous extension to $(0,0)$ by defining, for example,
\[
g(0,0) = 1,
\]
or any value in the interval $[1-\gamma,\,1]$.
With this extension, $f$ becomes continuous on $\mathbb{R}^2_+$,
and the trivial solution $(S_1,S_2) \equiv (0,0)$ is an equilibrium point.
The uniqueness and positivity of nontrivial solutions are not affected by this extension.
\end{remark}

\section{Proof of Theorem \ref{thm:thm_1}}\label{appendix:B}

\begin{proof}
The equilibrium points of system~\eqref{eq:equation_1} are obtained by setting
\[
\dfrac{dS_1}{dt} = \dfrac{dS_2}{dt} = 0,
\]
which yields the following system of algebraic equations:
\begin{align}
    S_1\psi_1\left(1 - \dfrac{1}{K}(S_1 + S_2) \right) \left(\dfrac{S_1+(1-\gamma)S_2}{S_1+S_2} \right) - \delta_1 S_1 &= 0, \label{eq:primeiraeq} \\
    S_2\psi_2 \left(1 - \dfrac{1}{K}(S_1 + S_2) \right) - \delta_2 S_2 &= 0. \label{eq:segundaeq}
\end{align}

From equation~\eqref{eq:segundaeq}, it follows that either $S_2 = 0$ or
\[
\psi_2 \left(1 - \dfrac{1}{K}(S_1 + S_2) \right) = \delta_2.
\]
In the latter case, we obtain
\begin{equation}\label{eq:sestrela}
    S^* = S_1 + S_2 = K\left(1 - \dfrac{\delta_2}{\psi_2} \right),
\end{equation}
which allows us to express $S_2$ as a function of $S_1$:
\[
S_2 = K\left(1 - \dfrac{\delta_2}{\psi_2} \right) - S_1.
\]

By substituting this expression into equation~\eqref{eq:primeiraeq} and solving for $S_1$, while introducing $R_1 = \frac{\psi_1}{\delta_1}$ and $R_2 = \frac{\psi_2}{\delta_2}$, assumed to be greater than $1$ under conditions~\eqref{eq:equation_2}-\eqref{eq:equation_3}, we obtain

\[
\hat{S}_1 = \left(1 - \frac{1}{R_2} \right)\frac{K}{\gamma} \left( \frac{R_2}{R_1} - (1 - \gamma) \right),
\]
from which we obtain
\[
\hat{S}_2 = K\left(1 - \dfrac{1}{R_2} \right) - \hat{S}_1,
\]
or equivalently,
\[
\hat{S}_2 = \left(1 - \frac{1}{R_2} \right)K\left[ 1 - \frac{1}{\gamma} \left( \frac{R_2}{R_1} - (1 - \gamma) \right) \right].
\]

Therefore, the coexistence equilibrium point is given by $E_3 = (\hat{S}_1, \hat{S}_2)$. 
For this point to exist, both coordinates must be positive. 
Hence, $\hat{S}_1 > 0$ if and only if $\frac{R_2}{R_1} > (1 - \gamma)$, 
and $\hat{S}_2 > 0$ if and only if $\frac{R_2}{R_1} < 1$. 
Consequently, the coexistence equilibrium $E_3$ exists if and only if
\[
1 - \gamma < \frac{R_2}{R_1} < 1.
\]

Next, we analyze the cases in which one of the populations goes extinct:

\begin{itemize}
    \item If $S_2 = 0$, equation~\eqref{eq:primeiraeq} reduces to
    \[
    S_1\psi_1 \left(1 - \dfrac{S_1}{K} \right) - \delta_1 S_1 = 0,
    \]
    which yields $S_1^* = 0$ or $S_1^* = K\left(1 - \dfrac{1}{R_1} \right)$, provided that $\psi_1 > \delta_1$.

    \item If $S_1 = 0$, from equation~\eqref{eq:sestrela} we obtain
    \[
    S_2^* = K\left(1 - \dfrac{1}{R_2} \right),
    \]
    with $\psi_2 > \delta_2$.
\end{itemize}

Therefore, the three equilibrium points of the system can be explicitly written in terms of the model parameters as follows:
\begin{itemize}
    \item $E_1 = \left(K\left(1 - \dfrac{1}{R_1} \right), 0\right)$,
    \item $E_2 = \left(0, K\left(1 - \dfrac{1}{R_2} \right)\right)$,
    \item $E_3 = (\hat{S}_1, \hat{S}_2)$, where
    \begin{align*}
        \hat{S}_1 &= \left(1 - \frac{1}{R_2} \right)\frac{K}{\gamma} \left( \frac{R_2}{R_1} - (1 - \gamma) \right), \\
        \hat{S}_2 &= \left(1 - \frac{1}{R_2} \right)K\left[ 1 - \frac{1}{\gamma} \left( \frac{R_2}{R_1} - (1 - \gamma) \right) \right].
    \end{align*}
\end{itemize}

We now proceed to analyze the local stability of the three equilibrium points obtained. For this purpose, we consider the Jacobian matrix of system \eqref{eq:equation_1}, which is given by \eqref{jacobi1}-\eqref{jacobi2}.

\begin{itemize}
    \item For the \textit{Wolbachia}-free equilibrium $E_1 = \left(K\left(1 - \dfrac{1}{R_1} \right), 0\right)$, we have:
$$
J(E_1) = \begin{bmatrix}
\delta_1 - \psi_1 & \delta_1(2-\gamma) - \psi_1 \\
0 & \dfrac{\psi_2}{R_1} - \delta_2
\end{bmatrix}.
$$
The eigenvalues of the characteristic polynomial $p(\lambda) = \det(J(E_1) - \lambda I)$ are given by
$\lambda_1 = \delta_1 - \psi_1$ and $\lambda_2 = \dfrac{\psi_2}{R_1} - \delta_2$.
Note that $\lambda_1 < 0$ due to the condition \eqref{eq:equation_2} imposed on the model. Therefore, for the equilibrium $E_1$ to be locally asymptotically stable, it is necessary that $\lambda_2 < 0$, which leads to the condition
$$
R_1 > R_2.
$$

\item For the wild-type-free equilibrium $E_2 = \left(0, K\left(1 - \dfrac{1}{R_2} \right)\right)$, we obtain:
$$
J(E_2) = \begin{bmatrix}
\psi_1(1-\gamma)\dfrac{1}{R_2} - \delta_1 & 0 \\
\delta_2 - \psi_2 & \delta_2 - \psi_2
\end{bmatrix}.
$$
In this case, the eigenvalues of the characteristic polynomial $p(\lambda) = \det(J(E_2) - \lambda I)$ are
$\lambda_1 = \psi_1(1-\gamma)\dfrac{1}{R_2} - \delta_1$ and $\lambda_2 = \delta_2 - \psi_2$.
As in the previous case, $\lambda_2 < 0$ by condition \eqref{eq:equation_2}. Hence, $E_2$ will be locally asymptotically stable if $\lambda_1 < 0$, that is, if
$$
R_2 > (1 - \gamma)R_1.
$$
\item To analyze the local stability of the coexistence equilibrium, we consider the trace and determinant of the Jacobian matrix 
\[
J(E_3) =
\begin{bmatrix}
\displaystyle \left.\frac{\partial f_1}{\partial S_1}\right|_{E_3} &
\displaystyle \left.\frac{\partial f_1}{\partial S_2}\right|_{E_3} \\[10pt]
\displaystyle \left.\frac{\partial f_2}{\partial S_1}\right|_{E_3} &
\displaystyle \left.\frac{\partial f_2}{\partial S_2}\right|_{E_3}
\end{bmatrix}.
\]

Let the equilibrium point be \(E_3 = (\hat{S}_1, \hat{S}_2)\), where
\begin{align*}
 \hat{S}_1 &= \left(1 - \frac{1}{R_2} \right)\frac{K}{\gamma} \left( \frac{R_2}{R_1} - (1 - \gamma) \right), \\
        \hat{S}_2 &= \left(1 - \frac{1}{R_2} \right)K\left[ 1 - \frac{1}{\gamma} \left( \frac{R_2}{R_1} - (1 - \gamma) \right) \right].
\end{align*}

Then, the entries of the Jacobian matrix evaluated at \(E_3\) are given by
\begin{align*}
\left.\frac{\partial f_1}{\partial S_1}\right|_{E_3} &= 
\psi_1\left(1 + \frac{c + \gamma - 1}{\gamma}\right)
- \frac{\psi_1}{\gamma}\left(1 - \frac{1}{R_2}\right)(c + \gamma - 1)(2 - c)
- \delta_1,  \\
\left.\frac{\partial f_1}{\partial S_2}\right|_{E_3} &= 
-\frac{\delta_2}{\psi_2}\frac{(c + \gamma - 1)^2}{\gamma}\psi_1
- \left(1 - \frac{1}{R_2}\right)\frac{(c + \gamma - 1)\psi_1}{\gamma},\\
\left.\frac{\partial f_2}{\partial S_1}\right|_{E_3} &= 
-\frac{(\psi_2 - \delta_2)(1 - c)}{\gamma}, \\[6pt]
\left.\frac{\partial f_2}{\partial S_2}\right|_{E_3} &=
(\psi_2 - \delta_2) - \frac{K}{\gamma}(\psi_2 - \delta_2)(1 + \gamma - c),
\end{align*}
where \(c = \dfrac{R_2}{R_1}\).

Hence, the determinant of the Jacobian matrix at \(E_3\) is
\[
\det(J(E_3)) = 
\frac{(\psi_2-\delta_2)
(\delta_1 \psi_2 - \delta_2 \psi_1)
(\delta_1 \psi_2 + \delta_2 \gamma \psi_1 - \delta_2 \psi_1)}{\delta_2\gamma \psi_1 \psi_2}.
\]

For the equilibrium \(E_3\) to be locally asymptotically stable, the following conditions must hold:
\[
\operatorname{tr}(J(E_3)) < 0 
\quad \text{and} \quad 
\det(J(E_3)) > 0.
\]

However, the existence of \(E_3\) requires that
\[
1 - \gamma < \dfrac{R_2}{R_1}< 1.
\]
Since \(\psi_2 > \delta_2\), we have
\begin{align*}
\delta_1 \psi_2 - \delta_2 \psi_1 &= \psi_1 \delta_2 (c - 1),\\
\delta_1\psi_2+\delta_2\gamma\psi_1 - \delta_2 \psi_1 &= \psi_1 \delta_2 (c + \gamma - 1).
\end{align*}
We can rewrite,
\[
\det(J(E_3)) = 
\frac{(\psi_2-\delta_2)(c-1)(c+\gamma-1)\psi_1\delta_2}{\gamma \psi_2} < 0,
\]
since \(c - 1 < 0\) by the existence condition of the equilibrium point.

Consequently, if the coexistence equilibrium \(E_3\) exists, then \(\det(J(E_3)) < 0\), and thus \(E_3\) is a saddle point (unstable).
\end{itemize}
\end{proof}

\section{Proof of Theorem \ref{thm:thm_3}}\label{appendix:C}

\begin{proof}
Let, 
\begin{equation*}
    B(t) = \dfrac{c \Bar{K} e^{r_2 t}}{ce^{r_2 t}-1},
\end{equation*}be the solution of 
\begin{eqnarray*}
    \dfrac{dB}{dt}(t) = B\psi_2\left(1 - \dfrac{B}{K}\right) - \delta_2 B 
\quad t \neq (n + l)\tau, \text{ } t \neq (n + 1)\tau,
\end{eqnarray*}the first equation of \eqref{eq:equation_10}, where $r_2 = \psi_2-\delta_2$, $\bar{K} = \frac{Kr_2}{\psi_2}$ and $c \in \mathbb{R}$ is a constant determined by the initial conditions of the problem. For $t = n\tau$ and $t = (n+l)\tau$, we denote the initial values at these times as $B(n\tau^+)$ and $B((n+l)\tau^+)$, respectively. The constant $c$ is then given by:
\begin{equation*}
c = \begin{cases}
\dfrac{B(n\tau^+) e^{-r_2 n\tau}}{B(n\tau^+) - \bar{K}}, & \text{for } t = n\tau, \\
\dfrac{B((n + l)\tau^+) e^{-r_2 (n + l)\tau}}{B((n + l)\tau^+) - \bar{K}}, & \text{for } t = (n + l)\tau.
\end{cases}
\end{equation*}
then,
\begin{align}\label{eq:equation_12}
    B(t) = \begin{cases}
        \dfrac{\bar{K} B(n\tau^+) e^{r_2 (t-n\tau)}}{B(n\tau^+)\left(e^{r_2 (t-n\tau)}     -1\right)+\bar{K}}, \quad  n\tau < t \leq (n+l)\tau,\\
        \dfrac{\bar{K} B((n+l)\tau^+) e^{r_2 (t-(n+l)\tau)}}{B((n+l)\tau^+)\left(e^{r_2 (t-(n+l)\tau)}     -1\right)+\bar{K}}, \quad  (n+l)\tau < t \leq (n+1)\tau,        
    \end{cases}
\end{align}
is the solution of system \eqref{eq:equation_9} between the pulses. At the pulse moments, when $t = (n+l)\tau$ with $n\geq 0$, the second equation of system \eqref{eq:equation_9} yields the following difference equation: 
\begin{eqnarray}\label{eq:pulse1}\nonumber
    B((n+l)\tau^+) =&(1-\varphi )B((n+l)\tau^-)\\
    = &(1-\varphi)\dfrac{\bar{K} B(n\tau^+) e^{r_2 ((n+l)\tau-n\tau)}}{B(n\tau^+)\left(e^{r_2 ((n+l)\tau-n\tau)}     -1\right)+\bar{K}}\\\nonumber
    = &(1-\varphi)\dfrac{\bar{K} B(n\tau^+) e^{r_2 l \tau}}{B(n\tau^+)\left(e^{r_2 l\tau}     -1\right)+\bar{K}}. \nonumber
\end{eqnarray} Similarly, at the next pulse time $t = (n+1)\tau$, $n \geq 0$, the third equation of system \eqref{eq:equation_9} gives the difference equation:
\begin{eqnarray}\label{eq:pulse2}\nonumber
    B((n+1)\tau^+) = & B((n+1)\tau^-) + u_n\\
    = &\dfrac{\bar{K} B((n+l)\tau^+) e^{r_2 ((n+1)\tau-(n+l)\tau)}}{B((n+l)\tau^+)\left(e^{r_2 ((n+1)\tau-(n+l)\tau)}     -1\right)+\bar{K}} + u_n\\\nonumber
    = &\dfrac{\bar{K} B((n+l)\tau^+) e^{r_2(1-l)\tau}}{B((n+l)\tau^+)\left(e^{r_2(1-l)\tau}     -1\right)+\bar{K}} + u_n. \nonumber
\end{eqnarray}
Substituting equation \eqref{eq:pulse1} into \eqref{eq:pulse2}, we obtain the recursive relation between consecutive pulse states:
\begin{align*}
    B((n+1)\tau^+) =& \dfrac{\bar{K} \left((1-\varphi)\dfrac{\bar{K} B(n\tau^+) e^{r_2 l \tau}}{B(n\tau^+)\left(e^{r_2 l\tau}     -1\right)+\bar{K}}\right) e^{r_2(1-l)\tau}}{\left((1-\varphi)\dfrac{\bar{K} B(n\tau^+) e^{r_2 l \tau)}}{B(n\tau^+)\left(e^{r_2 l\tau)}     -1\right)+\bar{K}}\right)\left(e^{r_2(1-l)\tau}     -1\right)+\bar{K}} + u_n,\\
    =&\frac{\bar{K}e^{r_2\tau}(1-\varphi)B(n\tau^+)}{B(n\tau^+)\left((1-\varphi)e^{r_2\tau} + \varphi e^{r_2l\tau}-1\right)+\bar{K}} +u_n.
\end{align*}
We can rewrite it as
\begin{equation}\label{eq:equation_15}B^{n+1} = h(B^n) := \frac{\bar{K}e^{r_2\tau}(1-\varphi)B^n}{B^n\left((1-\varphi)e^{r_2\tau} + \varphi e^{r_2l\tau}-1\right)+\bar{K}} + u_n,
\end{equation}
with $B^n := B(n\tau^+)$, which is the stroboscopic map of \eqref{eq:equation_9}. Set,
\begin{align*}
    h(B) = \frac{\bar{K}e^{r_2\tau}(1-\varphi)B}{B\left((1-\varphi)e^{r_2\tau} + \varphi e^{r_2l\tau}-1\right)+\bar{K}} +u_n,
\end{align*}Note that equation \eqref{eq:equation_15} has a single positive fixed point, since the equilibrium condition
\begin{align*}
h(B) = B,
\end{align*}
which determines the $\tau$-periodic fixed point of the recursive equation \eqref{eq:equation_15}, yields a unique positive solution through algebraic manipulation. Then,
\begin{equation}\label{eq:equation_16}
     A B^2 - \left(\bar{K}\left((1-\varphi)e^{r_2\tau} -1\right) + u_nA\right)B - u_n\bar{K} = 0,
\end{equation}where
\begin{equation}\label{eq:A_def}
    A := (1-\varphi)e^{r_2\tau} + \varphi e^{r_2 l \tau} - 1,
\end{equation}
and the discriminant \(\Delta\) is
\begin{equation*}\label{eq:Delta_def}
    \Delta := \big(\bar{K}\big((1-\varphi)e^{r_2\tau} - 1\big) + u_n A\big)^2 + 4 u_n \bar{K} A.
\end{equation*}
Since \(A> 0\) and \(\Delta>0\), the unique positive solution of \eqref{eq:equation_16} is
\begin{equation}\label{eq:B_plus}
    B^{+}
    \;=\;
    \dfrac{\bar{K}\big((1-\varphi)e^{r_2\tau} - 1\big) + u_n A + \sqrt{\big(\bar{K}\big((1-\varphi)e^{r_2\tau} - 1\big) + u_n A\big)^2 + 4 u_n \bar{K} A}}{2A}.
\end{equation}

Now, substituting $B^+$ into the first equation of \eqref{eq:equation_12}, we obtain  
\begin{align*}
    \Bar{B}(t) &= \dfrac{\bar{K} B^+ e^{r_2 (t-n\tau)}}{B^+\left(e^{r_2 (t-n\tau)}     -1\right)+\bar{K}}, \quad  n\tau < t \leq (n+l)\tau,
\end{align*}and for $(n + l)\tau < t \leq (n + 1)\tau$, we have \begin{align*}
    \Bar{B}(t) &= \dfrac{\bar{K} B^* e^{r_2 (t-(n+l)\tau)}}{B^*\left(e^{r_2 (t-(n+l)\tau)}     -1\right)+\bar{K}},
\end{align*}where $B^* = \dfrac{(1-\varphi)\bar{K}B^+e^{r_2l\tau}}{B^+\left(e^{r_2l\tau}-1\right)+\bar{K}}.$ Thus, the function 
\begin{align*}\label{eq:equation_18}
    \bar{B}(t) = \begin{cases}
        \dfrac{\bar{K} B^+e^{r_2 (t-n\tau)}}{B^+\left(e^{r_2 (t-n\tau)}     -1\right)+\bar{K}}, \quad  n\tau < t \leq (n+l)\tau,\\
        \dfrac{\bar{K} B^* e^{r_2 (t-(n+l)\tau)}}{B^*\left(e^{r_2 (t-(n+l)\tau)}     -1\right)+\bar{K}}, \quad  (n+l)\tau < t \leq (n+1)\tau,      \end{cases}
\end{align*}
defines the unique positive $\tau$-periodic solution of system \eqref{eq:equation_9}.

Furthermore, the function $h(B)$ is strictly increasing for $B > 0$, since it is composed of rational expressions with positive coefficients. Moreover,
\begin{equation*}
\lim_{B \to 0^+} h(B) = u_n > 0 \quad \text{and} \quad \lim_{B \to \infty} h(B) = \dfrac{\bar{K}e^{r_2\tau}(1-\varphi)}{(1-\varphi)e^{r_2\tau} + \varphi e^{r_2l\tau} - 1} + u_n.
\end{equation*}
Hence, $h(B)$ intersects the identity line $h(B) = B$ only once at $B = B^+$, and we conclude that:
\begin{align*}
h(B) < B &\quad \text{for all } B > B^+, \\
h(B) > B &\quad \text{for all } B < B^+.
\end{align*}
This implies that the sequence $\{B^n\}$ defined recursively by \eqref{eq:equation_15} converges monotonically to $B^+$ from any initial condition $B^0 > 0$. Thus, $B^+$ is globally asymptotically stable.
\begin{figure}[H]
\centering
\begin{tikzpicture}
\begin{axis}[
    axis lines = left,
    xlabel = $B$,
    ylabel = {$h(B)$},
    domain=0:10,
    samples=200,
    legend pos=north west,
    width=0.45\textwidth,   
    height=0.35\textwidth,  
    xtick={0,2,4,6,8,10},
    ytick={0,2,4,6,8,10},
    xmin=0, xmax=10,
    ymin=0, ymax=10,
    thick,
    grid=none
]

\def\a{4} \def\b{2.5} \def\c{3.5} \def\un{2}

\addplot[blue, thick, domain=0.1:10] {\a*\b*x / (x*\c + \a) + \un};
\addlegendentry{$h(B)$}

\addplot[red, dashed, domain=0:10] {x};
\addlegendentry{$h(B)=B$}

\pgfmathsetmacro{\A}{\c}
\pgfmathsetmacro{\Bcoeff}{\a*(1-\b) - \un*\c}
\pgfmathsetmacro{\Ccoeff}{- \un * \a}
\pgfmathsetmacro{\dis}{\Bcoeff*\Bcoeff - 4*\A*\Ccoeff}

\ifdim\dis pt<0pt
\else
  \pgfmathsetmacro{\xone}{(-\Bcoeff + sqrt(\dis))/(2*\A)}
  \pgfmathsetmacro{\xtwo}{(-\Bcoeff - sqrt(\dis))/(2*\A)}
  \pgfmathsetmacro{\zroot}{(\xone>\xtwo) ? \xone : \xtwo}
  \pgfmathsetmacro{\zroot}{(\zroot>0) ? \zroot : ((\xone>0)?\xone:\xtwo)}
  \addplot[only marks, mark=*, mark size=1.8pt] coordinates {(\zroot,\zroot)};
  \node[above right] at (axis cs:\zroot,\zroot) {$B^+$};
\fi
\end{axis}
\end{tikzpicture}
\caption{Illustration of the stroboscopic map $h(B)$ and the identity line $h(B)=B$. The unique intersection point $B^+$ is globally attractive.}

\end{figure}
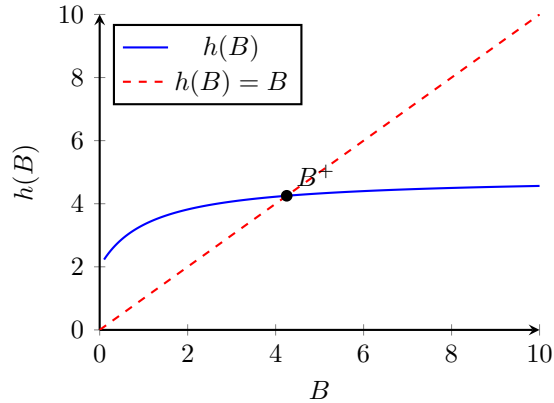
Since $B^+$ is the unique globally attractive fixed point of the stroboscopic map, and the solution $\bar{B}(t)$ is constructed from it, we conclude that any solution of \eqref{eq:equation_9} with $B(0^+) > 0$ converges to $\bar{B}(t)$ as $t \to \infty$.
\end{proof}

\section*{Acknowledgments}
The authors gratefully acknowledge the late Professor Sergio Muniz Oliva Filho for his guidance and support in previous works, which were fundamental to the development of the present study.

This project receives financial support from CAPES through STIC AMSUD (88881.878875/2023-01). JCSA thanks CAPES (Finance Code 001) for the scholarship. 
CES acknowledges FEEI-CONACYT-BIOCIVIP-AMSU99-7 and ARASY-ESTR01-23.


\bibliographystyle{unsrt}  
\bibliography{bibliografia.bib}

\end{document}